\documentclass{article}
\usepackage{style} 
\usepackage[style=alphabetic,maxalphanames=4,minalphanames=4]{biblatex}

\bibliography{main.bib}

\DeclareMathOperator{\supp}{supp}

\title{An (almost) efficient classical algorithm for sampling from typical Gibbs states}
\author{Saleh Naghdi\thanks{Institute for Quantum Information and Matter, Caltech; e-mail: {\href{mailto:mnaghdi@caltech.edu}{\texttt{mnaghdi@caltech.edu}}}} \and Eric R.\ Anschuetz\thanks{Department of Physics, National University of Singapore; Department of Computer Science, National University of Singapore; Institute for Quantum Information and Matter, Caltech; Walter Burke Institute for Theoretical Physics, Caltech; e-mail: {\href{mailto:e-ans@nus.edu.sg}{\texttt{e-ans@nus.edu.sg}}}}}
\date{}

\begin{document}

\maketitle
\begin{abstract}
Gibbs state preparation has attracted significant interest as a potential route to demonstrating quantum computational advantage in the study of physical systems. Although quantum advantage is known for computing local properties of worst-case Gibbs states, recent results drawing on spin glass theory suggest a different picture for average-case local systems: at temperatures where their Gibbs states can be efficiently prepared quantumly, quasi-polynomial time classical algorithms can also compute local expectation values. These results, however, leave open the possibility of an advantage under a stronger notion of simulation: \emph{sampling}. Gibbs sampling has been used as a subroutine in Boltzmann machine-based quantum learning algorithms, quantum supremacy proposals, and classifying phases in quantum many-body physical systems, though its classical average-case complexity for physically relevant models remains unresolved.

Here, we give strong evidence that there is \emph{no} exponential quantum advantage in Gibbs sampling for an ensemble of local systems known as the quantum $p$-spin model. We describe a classical algorithm for sampling from the Gibbs state of this model to any polynomially small error in total variation distance. Our algorithm is based on a meta-algorithm known as algorithmic stochastic localization (ASL), which was recently used to construct a polynomial-time sampling algorithm for the classical Sherrington--Kirkpatrick model to temperatures down to the spin glass phase transition. Our main technical result is a generalized convergence theorem for ASL which reduces approximate sampling from any distribution on the hypercube to estimating the mean of a related ``tilted distribution'' to sufficiently low additive error. We then extend known quasi-polynomial time algorithms for computing local expectation values of quantum spin glass systems to give a classical algorithm for computing these ``tilted means'' for the quantum $p$-spin model. Finally, we give a non-rigorous physics argument that this mean-estimation algorithm works down to the spin glass transition temperature of the model. As the spin glass transition is precisely where efficient quantum algorithms also fail, this suggests that efficient quantum algorithms cannot Gibbs sample the quantum $p$-spin model at temperatures below that for which classical algorithms are also (almost) efficient.
\end{abstract}

\tableofcontents

\section{Introduction}

Quantum simulation is widely touted as one of the most promising use-cases for quantum computers in the coming decades~\cite{Feynman_1982,Lloyd_1996,McArdle_2020}. One setting which has attracted particular attention in recent years is the study of low-temperature properties of quantum many-body systems. This focus is partly driven by the industrial relevance of this task---many problems in chemistry, biology, and physics can be considered as special cases of this general problem---but also driven by theoretical evidence of there being a quantum-classical separation in the computational complexity of this task. In particular, conditioned on widely-believed complexity theoretic assumptions, it is known that there exist ``worst case'' instances for which estimating local observables in low-temperature quantum thermal states is classically hard yet quantumly easy~\cite{rouze2025efficient,chen2024local}.

That said, just because a problem is difficult in the worst case does not mean it is difficult in practice. For instance, the Sachdev--Ye--Kitaev (SYK) model is an ensemble of local fermionic quantum many-body systems for which, for worst-case draws of the system from the ensemble, is expected to be classically (and quantumly) hard to estimate local properties of its low-energy space~\cite{PhysRevLett.98.110503}. However, recent work has shown that with high probability over draws from this ensemble, there is a classical quasi-polynomial time algorithm for computing local expectation values of Gibbs states~\cite{zlokapa2026sykthermalexpectationsclassically,zlokapa2026rigorousquasipolynomialtimeclassicalalgorithm}. These algorithms are part of a class of algorithms for estimating local observables based on Barvinok's interpolation method~\cite{v011a013}, and for disordered systems are widely expected to work at temperatures down to the \emph{spin glass phase transition}~\cite{10.1145/3357713.3384322}---namely, to temperatures where this problem also becomes quantumly difficult~\cite{gamarnik2024slowmixingquantumgibbs,rakovszky2026bottlenecks,anschuetz2025efficientlearningimpliesquantum,zlokapa2025averagecasequantumcomplexityglassiness}. This suggests that for typical local Hamiltonians drawn from this ensemble, there is no exponential quantum advantage in estimating local properties of low-temperature physical systems; classical algorithms can efficiently probe the same temperatures that quantum algorithms can.

This previous line of work leaves open whether quantum algorithms are able to outperform classical algorithms under a stronger notion of simulation: \emph{sampling}. While this is perhaps less natural from a physical perspective, many quantum algorithms use sampling from thermal states as a subroutine, including Boltzmann machine-based learning algorithms~\cite{Amin_2018,Kieferova_2017,wiebe2019quantumlanguageprocessing}, quantum supremacy proposals~\cite{10756075,Rajakumar_2026}, and classifying phases in quantum many-body physical systems~\cite{Chiu_2019,Keesling_2019,Zhang_2025,Manovitz_2025}.

The range of classical algorithms for Gibbs sampling of quantum Hamiltonians remains limited. The main approach, quantum Monte Carlo, is efficient for special classes of Hamiltonians, such as commuting or stoquastic systems, for which the infamous \textit{sign problem} is absent \cite{crosson_rapid_2018,krzakala_path-integral_2008}. This breaks down for generic noncommuting, nonstoquastic Hamiltonians. A recent result gives an efficient classical sampler for noncommuting Hamiltonians at sufficiently high temperature~\cite{yin_polynomial-time_2023}, but no comparable general algorithm is known in the low-temperature regimes.

Here, we give strong evidence that there is \emph{not} a quantum advantage in sampling from the Gibbs state of typical local Hamiltonians. As a concrete example of an ensemble of local systems, we consider the \emph{quantum $p$-spin model}:
\begin{equation}
    H = \sum_{I \in \binom{[N]}{p}\times {\{X,Y,Z\}^p}} J_I \hat{\sigma}_I,
\end{equation}
where the $J_I$ denote i.i.d.\ Gaussian random variables normalized such that $H$ has operator norm $\operatorname{\Theta}(N)$ w.h.p.\ and the sum runs over all $p$-local Pauli operators. We then consider a meta-algorithm known as \emph{algorithmic stochastic localization} (ASL), which has seen much recent success as a classical technique for sampling from classical disordered systems in their paramagnetic phase~\cite{9996629,davies2026potentialhessianascentiii,lee2026potentialhessianascentiv}. We prove a general convergence theorem for this meta-algorithm to reduce the problem of sampling from a Gibbs state to estimating local observables in a ``tilted'' Gibbs state. We then show that recent quasi-polynomial time algorithms based on Barvinok's interpolation method~\cite{v011a013,10.1145/3357713.3384322,bencs_zeros_2026,zlokapa2026sykthermalexpectationsclassically,zlokapa2026rigorousquasipolynomialtimeclassicalalgorithm} can be extended to compute local observables of these tilted systems. While---as is the case for all papers in this line of work---we are unable to compute the exact temperatures Barvinok's method can probe, we give physics-based evidence that this tilted mean estimation algorithm succeeds at any temperature above the spin glass phase transition. As the spin glass phase transition temperature is exactly where the problem also becomes quantumly difficult~\cite{gamarnik2024slowmixingquantumgibbs,rakovszky2026bottlenecks,anschuetz2025efficientlearningimpliesquantum,zlokapa2025averagecasequantumcomplexityglassiness}, this suggests that not only does no exponential quantum advantage exist for estimating local expectation values, but also none exists for sampling.

\subsection{Informal Presentation of Technical Results}

\subsubsection{General Convergence of Algorithmic Stochastic Localization}

Our main technical result is a reduction from Gibbs sampling for \emph{any} system to computing expectation values of a ``tilted'' Gibbs distribution. Formally, we show that the \emph{stochastic localization} process introduced in \cite{eldan2020taming,eldan2022109392} can be algorithmized under very general conditions.

Before we continue, we briefly describe the stochastic localization process and its interesting properties. Let $p\left(\bm{z}\right)$ be an arbitrary distribution on the $N$-dimensional hypercube, and consider the \emph{tilted distribution} for any $\bm{y}\in\mathbb{R}^N$:
\begin{equation}\label{eq:tilt_dist}
    p_{\bm{y}}\left(\bm{z}\right)\propto\exp\left(\bm{y}^\intercal\cdot\bm{z}\right)p\left(\bm{z}\right)
\end{equation}
with mean:
\begin{equation}
    \bm{m}_{\bm{y}}=\sum_{\bm{z}\in\left\{-1,1\right\}^N}p_{\bm{y}}\left(\bm{z}\right)\bm{z}.
\end{equation}
Finally, let $\dd{\bm{B}_t}$ denote standard Brownian motion. Then, the stochastic localization process is defined as:
\begin{equation}\label{eq:stoch_loc_informal}
    \dd{\bm{u}_t}=\bm{m}_{\bm{u}_t}\dd{t}+\dd{\bm{B}_t},\quad\bm{u}_0=0
\end{equation}
over $t\in\mathbb{R}_{\geq 0}$. What makes stochastic localization interesting is that, in the $t\to\infty$ limit, $\bm{m}_{\bm{u}_t}$ almost surely localizes to a point $\bm{m}_\infty$ on the hypercube; furthermore, over the randomness of the Brownian motion, $\bm{m}_\infty$ is distributed according to the untilted distribution $p$. These properties can be seen either through direct computation (see \cite{eldan2022109392}) or via a Bayesian interpretation of the stochastic localization process \cite[Section~4]{AFST_2023_6_32_5_939_0}. These properties suggest that if one were able to:
\begin{enumerate}
    \item exactly compute $\bm{m}_{\bm{y}}$ at any $\bm{y}$, and
    \item simulate Eq.~\eqref{eq:stoch_loc_informal} exactly for ``infinite time,''
\end{enumerate}
then one would be able to produce samples from $p\left(\bm{z}\right)$.

Of course, these two requirements are impossible to satisfy for any algorithm operating at finite precision. However, in certain settings---particularly, for sampling from the classical Sherrington--Kirkpatrick (SK) model---it has been shown that stochastic localization can be \emph{algorithmized} into what is known as algorithmic stochastic localization (ASL)~\cite{9996629,davies2026potentialhessianascentiii,lee2026potentialhessianascentiv}. The tilted and untilted SK models have the nice property that their means can be described as the stationary point of what is known as the TAP free energy~\cite{Thouless_1977}, which in turn can be efficiently approximated outside of the spin glass phase using an algorithm known as approximate message passing (AMP)~\cite{doi:10.1137/20M132016X,chen2023}. This characterization allows the authors of \cite{9996629,davies2026potentialhessianascentiii,lee2026potentialhessianascentiv} to show that for the SK model, outside of the spin glass phase, the stochastic localization process converges quickly and is robust to the error in using AMP to estimate $\bm{m}_{\bm{y}}$.

However, no analogue of AMP is known to exist for quantum many-body systems. To circumvent this, we use a trick involving Girsanov's Theorem to show that ASL converges for \emph{any} choice of algorithm for estimating local expectation values of the tilted distribution, as long as it is sufficiently accurate. This results in our main technical result, proven in Section~\ref{sec:asl_conv}.
\begin{theorem}[Convergence of ASL, informal version of Theorem~\ref{thm:num_guar_sl}]\label{thm:inf_num_guar_sl}
Let $p$ be a distribution on $\{\pm1\}^N$, and let $\bm{m}(\bm u)$ denote the mean of its tilted distribution (Eq.~\eqref{eq:tilt_dist}) at tilt $\bm u\in\mathbb R^N$. Let $\hat{\bm{m}}:\mathbb R^N\to[-1,1]^N$ be an estimator with the error guarantee:
\begin{align}
    \frac{1}{\sqrt{N}}\sup_{\bm u\in\mathbb R^N}\norm{\hat{\bm{m}}(\bm u)-\bm{m}(\bm u)}_2\leq \rho.
\end{align}
Define the ASL process for a step size $\delta$ as
\begin{align}
    \hat{\bm u}_{k+1}
    =\hat{\bm u}_k+\delta\,\hat{\bm{m}}(\hat{\bm u}_k)+\bm{W}_k,
    \qquad \hat{\bm u}_0=0,
\end{align}
where $\bm{W}_k\sim\mathcal{N}\left(0,\delta\right)^{\otimes N}$ i.i.d. Let $R:[-1,1]^N\to\{\pm1\}^N$ denote coordinate-wise randomized rounding. Then, for every $\epsilon>0$, there exist
\begin{align}
    T=\poly(N,1/\epsilon),\qquad
    \delta^{-1}=\poly(N,1/\epsilon),\qquad
    \rho^{-1}=\poly(N,1/\epsilon)
\end{align}
such that, with $K=T/\delta=\poly(N,1/\epsilon)$,
\begin{align}
    \TVD\!\left(
        p,\,
        \mathcal L\!\left(R\!\left(\hat{\bm{m}}(\hat{\bm u}_K)\right)\right)
    \right)
    \leq \epsilon.
\end{align}
\end{theorem}

\paragraph{Prior work in the diffusion modeling literature.}

After the first draft of this manuscript was made public, it was brought to the authors' attention that similar progress on algorithmic stochastic localization has been made in the diffusion modeling literature, where the study of stochastic localization is motivated by the fact that diffusion models can be viewed as a stochastic localization process in $\mathbb{R}^N$~\cite{montanari2025samplingdiffusionsstochasticlocalization}. In particular, as in the spin glass literature~\cite{9996629,davies2026potentialhessianascentiii,lee2026potentialhessianascentiv}, initial progress in understanding the convergence of diffusion models assumed the drift of the process was a Lipschitz function of the tilt~\cite{montanari2024posteriorsamplinghighdimension,kellermann2026totalvariationguaranteessampling}, or otherwise assumed that the distribution satisfied a log-Sobelev inequality~\cite{3600270.3601932}. Motivated by this prior work, it was then later shown that these assumptions could be circumvented when the distribution is assumed to have compact support through either a change-of-measure technique or through the use of Girsanov's Theorem~\cite{lee2022convergence,chen2023sampling}. Our main result differs in two ways. First, the stochastic localization process considered in the diffusion modeling literature differs from the process considered both here and in the spin glass literature~\cite{9996629,davies2026potentialhessianascentiii,lee2026potentialhessianascentiv}. In the former, the drift of the stochastic localization process is related to the ``score function'' of the diffusion model, which in our setting is a linear function of both the tilted mean and the tilt itself. Our choice of stochastic localization process enables us to substantially improve the cost of our sampling algorithm with respect to $\epsilon$. For instance, \cite[Corollary~6]{chen2023sampling} requires $\widetilde{O}(N^5/\epsilon^{12})$ discretization steps---i.e., $\widetilde{O}(N^5/\epsilon^{12})$ calls to the tilted mean estimator $\hat{\bm{m}}$---to sample from a distribution $\epsilon$-close in Wasserstein distance to the target distribution, whereas for the process we consider $O(N^5/\epsilon^8)$ steps suffice (see Eq.~\eqref{eq:choice_of_params_asl_w2}). Second, by considering a more constrained setting---we assume the distribution of interest has support on the hypercube---we are able to improve the convergence in Wasserstein distance demonstrated in \cite{lee2022convergence,chen2023sampling} to a convergence in total variation distance. We achieve this via a randomized rounding of the tilted means the stochastic localization process produces, allowing us to show that $O(N^5/\epsilon^4)$ discretization steps suffice to achieve $\epsilon$ error in total variation distance (see Eq.~\eqref{eq:choice_of_params_asl}).

\subsubsection{Application to the Quantum $p$-Spin Model}

As an application of Theorem~\ref{thm:num_guar_sl}, we consider the problem of Gibbs sampling the quantum $p$-spin model~\cite{swingle_bosonic_2023}. This is an ensemble of local Hamiltonians that has previously been used to probe the average-case complexity of the ground state problem~\cite{Anschuetz_2025,cbqf-d24r,bhattacharya2026parisiformulagroundstate} as well as of Gibbs state preparation~\cite{basso2024optimizingrandomlocalhamiltonians,zlokapa2025averagecasequantumcomplexityglassiness}.
\begin{defn}[Quantum $p$-spin model]
    The \emph{quantum $p$-spin model} is an ensemble of local Hamiltonians:
   \begin{equation}
        \bm{H} = \sqrt{\frac{1}{3^p\binom{N-1}{p-1}}}\sum_{I \in \binom{[N]}{p}\times {\{X,Y,Z\}^p}} J_I \hat{\sigma}_I;
    \end{equation}
    the sum runs over all $p$-local Pauli operators, the $J_I$ are i.i.d.\ standard Gaussian random variables, and the normalization chosen such that $\left\lVert H\right\rVert=\operatorname{\Theta}\left(N\right)$ w.h.p.
\end{defn}
We here consider the problem of sampling from the thermal state of $\bm{H}$:
\begin{defn}[Gibbs sampling the quantum $p$-spin model]
    Given $\bm{J}$ and inverse temperature $\beta$, the task of \emph{Gibbs sampling to error $\epsilon$} is to produce samples from a distribution $q$ approximating:
    \begin{equation}\label{eq:gibbs_dist_intro}
        p\left(\bm{z}\right)=\frac{\bra{\bm{z}}e^{-\beta\bm{H}}\ket{\bm{z}}}{\Tr\left(e^{-\beta\bm{H}}\right)}
    \end{equation}
    to error $\epsilon$ total variation distance. Here, $\ket{\bm{z}}$ denotes a computational basis state, though by the rotational invariance of the quantum $p$-spin model our results hold for any choice of product state basis.
\end{defn}
By Theorem~\ref{thm:inf_num_guar_sl}, the classical complexity of this problem entirely hinges on constructing an estimator $\hat{\bm{m}}(\bm u)$ of the ``tilted mean'' $\bm{m}\left(\bm{u}\right)$, i.e., the mean of the tilted distribution:
\begin{equation}
    p_{\bm{u}}\left(\bm{z}\right)=\frac{e^{\bm{u}^\intercal\cdot\bm{z}}p\left(\bm{z}\right)}{\sum\limits_{\bm{z}\in\left\{-1,1\right\}^N}e^{\bm{u}^\intercal\cdot\bm{z}}p\left(\bm{z}\right)}.
\end{equation}
We show in Section~\ref{sec:quasi_poly_est_alg} that recently-developed quasi-polynomial time classical algorithms for computing thermal expectation values~\cite{10.1145/3357713.3384322,bencs_zeros_2026,zlokapa2026sykthermalexpectationsclassically,zlokapa2026rigorousquasipolynomialtimeclassicalalgorithm} can be adapted to estimate tilted means of the quantum $p$-spin model. These quasi-polynomial time algorithms are based on Barvinok's interpolation method~\cite{v011a013}, which essentially works by analytically continuing the partition function defined on the complex-$\beta$ plane in a zero-free region around $\beta=0$. The main difficulty in proving convergence, then, is proving the existence of such a zero-free region with high probability over the randomness of the ensemble. Physically, it is widely believed that this zero-free region corresponds to the high-temperature phase of the model~\cite{10.1145/3357713.3384322}. More recently, for the SYK model it was rigorously shown that a zero-free region exists at all sufficiently high constant temperature~\cite{zlokapa2026rigorousquasipolynomialtimeclassicalalgorithm}, though as non-rigorous physics heuristics show the SYK model has no finite-temperature spin glass phase transition this is not expected to be tight~\cite{zlokapa2026sykthermalexpectationsclassically}.

Here, we show these methods can be adapted in the following way.
\begin{prop}[Barvinok's applied to quantum $p$-spin, informal versions of Propositions~\ref{claim:mag_from_finite_difference} and~\ref{prop:barvinoks}]\label{prop:barv_inf}
    Define the tilted partition function:
    \begin{equation}
        Z_{\bm{y}}\left(\beta\right):=\Tr\left(\exp\left(\sum_{i=1}^N y_i\bm{\sigma}_i^Z\right)\exp\left(-\beta\bm{H}\right)\right).
    \end{equation}
    Take $\delta>0$ and $\beta^\ast>0$ such that w.h.p., for all $\bm{y}\in\mathbb{R}^N$, $Z_{\bm{y}}\left(\beta\right)$ is zero-free whenever $\Re\left\{\beta\right\}\leq\beta^\ast$ and $\Im\left\{\beta\right\}\leq\delta$. Then for any choice of $\epsilon>0$ there is an $N^{\operatorname{O}\left(\ln\left(N/\epsilon\right)\right)}$-time algorithm for estimating $\bm{m}\left(\bm{u}\right)$ to additive error $\epsilon$ in each component that works w.h.p. over the draw of $\bm{H}$.
\end{prop}
Just as other work in this area, we unfortunately are unable to rigorously show the zero freeness of $Z_{\bm{y}}$ down to the spin glass phase transition. Instead, we give a mixture of rigorous and physics-based evidence to support this conclusion, all shown in Section~\ref{sec:quasi_poly_est_alg}:
\begin{enumerate}
    \item We show that $\E Z_{\bm{y}}\left(\beta\right)$ is zero-free in a constant-radius disk around the origin.
    \item We show that $Z_{\bm{y}=\bm{0}}\left(\beta\right)$ concentrates, and therefore is zero-free in a constant-radius disk around the origin with high probability.
    \item By appealing to the classical spin glass literature, we argue why the inclusion of a nonzero tilt $\bm{y}$ should not destroy zero-freeness~\cite{bencs_zeros_2026}.
    \item Assuming Fisher zeros exactly correspond to phase transitions~\cite{10.1145/3357713.3384322}, we claim that that $Z_{\bm{y}}$ is indeed zero-free down to the spin glass phase transition. 
\end{enumerate}

Whatever inverse temperature $\beta^\ast$ one chooses to believe Barvinok's interpolation method succeeds in producing local expectation values, combining Theorem~\ref{thm:inf_num_guar_sl} and Proposition~\ref{prop:barv_inf} yields our main result.
\begin{corollary}[Gibbs sampling the quantum $p$-spin model, informal]
    Let $\beta^\ast$ be the inverse temperature at which the zero-freeness conditions of Proposition~\ref{prop:barv_inf} holds. Let $p_\beta$ be as in Eq.~\eqref{eq:gibbs_dist_intro}. For any choice of $\epsilon>0$, with high probability over the quantum $p$-spin ensemble, ASL with Barvinok interpolation as a subroutine has time complexity $N^{\operatorname{O}\left(\log\left(N/\epsilon\right)\right)}$ and produces samples from a distribution $q$ satisfying:
    \begin{equation}
        \TVD\left(p_\beta,q\right)\leq\epsilon.
    \end{equation}
\end{corollary}
If indeed $\beta^\ast$ corresponds to the spin glass phase transition for the quantum $p$-spin model, as the spin glass phase transition is widely believed to be where quantum open system dynamics exhibit mixing times exponential in $N$~\cite{sachdev2024quantumspinglassessachdevyekitaev,basso2024optimizingrandomlocalhamiltonians,zlokapa2026sykthermalexpectationsclassically}, this suggests there is no exponential quantum advantage in Gibbs sampling from the quantum $p$-spin model.

\subsection{Discussion}

Our results follow a long string of evidence that there is no exponential quantum advantage in computing static properties of physical systems~\cite{Lee_2023,zlokapa2026sykthermalexpectationsclassically,zlokapa2026rigorousquasipolynomialtimeclassicalalgorithm}, upgrading previous dequantization results in computing local expectation values of Gibbs states to dequantization of sampling from Gibbs states (at temperatures where the former is classically easy). We argue that rather than hoping for exponential separations in the computation of static properties of quantum many-body physical systems, one should instead look in one of two places for quantum advantage in quantum simulation: either settling for \emph{polynomial separations} in estimating static properties, or searching for an exponential advantage in computing \emph{dynamical properties}.

Whether there remain substantive polynomial quantum-classical separations in either estimating local observables or Gibbs sampling remains an open question. The current best-known algorithms for estimating local observables are based on Barvinok's interpolation method which inherently carries with it a quasi-polynomial overhead~\cite{zlokapa2026sykthermalexpectationsclassically,zlokapa2026rigorousquasipolynomialtimeclassicalalgorithm}. Physics heuristics suggest that iteratively solving the quantum Thouless--Anderson--Palmer (TAP) equations might be a rigorous way to achieve polynomial time complexity in estimating local observables~\cite{PhysRevB.64.014206}. For classical spin glass systems this has been made rigorous~\cite{bolthausen2014iterative}, and currently forms the basis for the best-known classical algorithms for computing local expectation values up to the spin glass phase transition~\cite{10.1002/cpa.21922,doi:10.1137/20M132016X}, though the rigor of the quantum analogue remains an open question. If it were made rigorous, that would tighten the gap between quantum and classical algorithms---indeed, our classical algorithm here would become polynomial time, rather than quasi-polynomial time---though even then a polynomial separation remains possible. Further progress in understanding the fine-grained average-case complexity of this problem is needed to make further progress on this question.

Another strategy in the search for quantum advantage is to think about dynamics problems rather than statics problem. This category includes questions such as understanding transport properties in physical systems, measuring some local property after appropriate time dynamics, and so forth. This includes recent experiments by Google in computing out-of-time-order correlators~\cite{GoogleQuantumAI_2025,king2025simplifiedversionquantumotoc2}, which remains classically hard. What makes dynamical properties different from static properties is that all currently-known classical algorithms for this problem require approximately implementing the time dynamics of the system. One potential way to begin to make this rigorous is to understand how the low-energy physics of these systems constrain algorithms which take this approach. For instance, previous work on the SYK model demonstrated that all low-energy states of the SYK model have lower-bounded classical descriptions, including bond dimension, circuit complexity, and neural network representation complexity~\cite{cbqf-d24r,Ding_2026}; this gives further evidence that computing dynamical properties of such states may retain a large quantum advantage. We hope to investigate this further in future work.

Our results reduce simulating sampling from a quantum system to estimating local observables in a related model. While we here focus on Gibbs sampling as an application of this general reduction, it would be interesting to consider other sampling problems for which this reduction might be useful. As an illustrative example as to why this technique cannot dequantize \emph{everything}, consider the problem of sampling from shallow-depth quantum circuits. More specifically, consider the ensemble of states $\ket{\psi}=\bm{U}\ket{\bm{0}}$ where $\bm{U}$ is a 2D shallow circuit of sufficiently large constant depth with local unitaries chosen Haar randomly. This distribution is conjectured to be difficult to sample from classically~\cite{PhysRevX.12.021021}. Algorithmic stochastic localization reduces the sampling from this distribution to estimating expectation values of the tilted distribution:
\begin{equation}
    p_{\bm{y}}\left(\bm{z}\right)=\frac{\exp\left(\bm{y}^\intercal\cdot\bm{z}\right)\left\lvert\bra{\bm{z}}\bm{U}\ket{\bm{0}}\right\rvert^2}{\sum\limits_{\bm{z}\in\left\{-1,1\right\}^N}\exp\left(\bm{y}^\intercal\cdot\bm{z}\right)\left\lvert\bra{\bm{z}}\bm{U}\ket{\bm{0}}\right\rvert^2}.
\end{equation}
While expectation values of the local observables are easy to compute when $\bm{y}=\bm{0}$ due to $\bm{U}$ being shallow---i.e., while the numerator is easy to compute---the tilted partition function $\sum_{\bm{y}}\exp\left(\bm{y}^\intercal\cdot\bm{z}\right)\left\lvert\bra{\bm{z}}\bm{U}\ket{\bm{0}}\right\rvert^2$ is generally difficult to estimate when $\bm{y}\neq\bm{0}$. We hope to investigate other sampling problems using our techniques in future work.

\section*{AI Usage Statement}

In the original version of this work, drafted pre-ChatGPT 5.6, we proved our main result (Theorem~\ref{thm:num_guar_sl}) modulo requiring an extra assumption on the Lipschitz continuity of the tilted mean estimator $\hat{\bm{m}}_{\bm{y}}$ with respect to the tilt $\bm{y}$. We prompted ChatGPT 5.6 to circumvent this Lipschitz continuity assumption, which led to it telling us the main idea behind the current proof of Theorem~\ref{thm:num_guar_sl}. This new strategy leverages a result in the theory of stochastic processes known as Girsanov's Theorem; ChatGPT's stated source for this idea was \cite{kellermann2026totalvariationguaranteessampling}, which proves convergence of an algorithmized version of stochastic localization in the special setting of a subroutine used in certain AI diffusion models.

More closely related in setting are the techniques used in the very recent paper \cite{lee2026potentialhessianascentiv}, which gives an efficient sampling algorithm for the classical Sherrington--Kirkpatrick (SK) model down to the spin glass phase transition. While this paper was posted publicly only after our query to ChatGPT, the idea for using this technique to prove the convergence of ASL originates with \cite{davies2026potentialhessianascentiii}, the preceding paper of the series culminating in \cite{lee2026potentialhessianascentiv}. There, the use of Girsanov's Theorem was stated as a possible proof technique but could not be used directly due to some technical conditions on the SK TAP equations that we do not run into in our analysis.

Finally, we also used ChatGPT 5.6 to check our proofs---we made minor corrections based on these discussions with the model.

\section*{Acknowledgments}
This manuscript was prepared in connection with the QIP 2027 submission deadline. We thank Alexander Zlokapa for many enlightening conversations regarding this work, and Sitan Chen for informing us of work in the diffusion modeling literature which also uses Girsanov's Theorem to demonstrate efficient convergence of a discretized stochastic localization process in that setting. We wish to thank Nadine Meister for their \LaTeX{} typesetting template used in this paper. E.R.A.\ was funded in part by the Walter Burke Institute for Theoretical Physics at Caltech. E.R.A.\ is supported by the National University of Singapore start-up grant A-0010665-00-00 and strategic hire fund A-0010665-01-00.

\section{Convergence of Algorithmic Stochastic Localization}\label{sec:asl_conv}

We here summarize the \emph{stochastic localization} process and its properties when applied to measures on the Boolean hypercube~\cite{eldan2020taming,eldan2022109392}. For a Bayesian interpretation of stochastic localization, we refer the reader to \cite[Section~4]{AFST_2023_6_32_5_939_0}.
\begin{prop}[Stochastic localization on the Boolean hypercube~{\cite[Proposition~10 and $v=0$ case of Proposition~11]{eldan2022109392}}]\label{prop:stoch_loc}
    Let $p$ be a probability distribution on $\left\{-1,1\right\}^N$. Let $\bm{B}_t$ be standard Brownian motion over $\mathbb{R}^N$. Define the \emph{tilted distribution} for any $\bm{y}\in\mathbb{R}^N$:
    \begin{equation}\label{eq:tilted_dist}
        p_{\bm{y}}\left(\bm{z}\right):=Z\left(\bm{y}\right)^{-1}\exp\left(\bm{y}^\intercal\cdot\bm{z}\right)p\left(\bm{z}\right),
    \end{equation}
    where:
    \begin{equation}
        Z\left(\bm{y}\right):=\sum_{\bm{z}\in\left\{-1,1\right\}^N}\exp\left(\bm{y}^\intercal\cdot\bm{z}\right)p\left(\bm{z}\right).
    \end{equation}
    Finally, let:
    \begin{equation}
        \bm{m}\left(\bm{y}\right)=\bm{m}_{\bm{y}}=\sum_{\bm{z}\in\left\{-1,1\right\}^N}p_{\bm{y}}\left(\bm{z}\right)\bm{z}
    \end{equation}
    denote the mean of the tilted distribution.

    The \emph{stochastic localization process}:
    \begin{equation}\label{eq:stoch_loc_process}
        \dd{\bm{u}_t}=\bm{m}\left(\bm{u}_t\right)\dd{t}+\dd{\bm{B}_t},\quad\bm{u}_0=0,
    \end{equation}
    over $t\in\mathbb{R}_{\geq 0}$ satisfies:
    \begin{enumerate}
        \item For all $A\subseteq\left\{-1,1\right\}^N$, $\left(p_{\bm{u}_t}\left(A\right):=\sum_{\bm{z}\in A}p_{\bm{u}_t}\left(\bm{z}\right)\right)_{t\geq 0}$ is a martingale.
         \item $\bm{m}_\infty:=\lim_{t\to\infty}\bm{m}\left(\bm{u}_t\right)$ exists.
         \item $\bm{m}_\infty\in\left\{-1,1\right\}^N$ almost surely.
         \item $\bm{m}_\infty$ is distributed according to $p$.
    \end{enumerate}
\end{prop}
To use stochastic localization as a sampling algorithm, one must then ensure that:
\begin{enumerate}
    \item There exists an efficient algorithm for computing means of the tilted distribution Eq.~\eqref{eq:tilted_dist}.
    \item The exact stochastic localization process Eq.~\eqref{eq:stoch_loc_process} converges quickly.
    \item The stochastic localization process is robust to errors induced by discretization of the time evolution and error in approximating the tilted means $\bm{m}_{\bm{y}}$.
\end{enumerate}
Condition 1 is problem-dependent and will be analyzed for the quantum $p$-spin model in Section~\ref{sec:quasi_poly_est_alg}. In the remainder of this section, we demonstrate Conditions 2 and 3 generally. In Section~\ref{sec:conv_rate_exact_stoch_loc}, we show that the exact stochastic localization process converges quickly in Wasserstein distance, which also bounds the convergence in total variation distance; this result is mostly adapted from prior work. Then, in Section~\ref{sec:stoch_loc_robust} we prove that stochastic localization is robust to errors in estimating $\bm{m}$ as well as errors induced by discretizing the process, and combine this result with the results of Section~\ref{sec:conv_rate_exact_stoch_loc} to prove our main convergence theorem (Theorem~\ref{thm:num_guar_sl}).

Our main technique for demonstrating Condition 3 will utilize Girsanov's Theorem, which has previously seen use in demonstrating the convergence of discretized stochastic localization processes in the diffusion modeling literature~\cite{chen2023sampling,kellermann2026totalvariationguaranteessampling}. By considering a slightly different stochastic localization process and specifying to distributions on the hypercube, we are able to improve the convergence of our algorithm in the target error $\epsilon$. For instance, \cite[Corollary~6]{chen2023sampling} requires $\widetilde{O}(N^5/\epsilon^{12})$ discretization steps---i.e., $\widetilde{O}(N^5/\epsilon^{12})$ calls to the tilted mean estimator $\hat{\bm{m}}$---to sample from a distribution $\epsilon$-close in Wasserstein distance to the target distribution, whereas for the process we consider we will see that $O(N^5/\epsilon^8)$ steps suffice (see Eq.~\eqref{eq:choice_of_params_asl_w2}). Furthermore, by implementing a randomized rounding procedure---taking advantage that the true distribution has support only on the hypercube---we are also able to show convergence in total variation distance, where we show that $O(N^5/\epsilon^4)$ discretization steps suffice to achieve $\epsilon$ error in total variation distance (see Eq.~\eqref{eq:choice_of_params_asl}).

\subsection{Convergence rate of the exact stochastic localization process}\label{sec:conv_rate_exact_stoch_loc}

We first prove that Condition 2 holds under very general conditions; our proof is adapted from \cite[Lemmas~III.3 and~VII.3]{9996629}, where this Condition was demonstrated for the Sherrington--Kirkpatrick model. In what follows we use $\mathcal{L}\left(\cdot\right)$ to denote the law of $\cdot$, $\operatorname{W}_2$ to denote the (classical) Wasserstein distance of order $2$, and $\TVD$ to denote the total variation distance.
\begin{prop}[Convergence rate of stochastic localization, adapted from~{\cite[Lemmas~III.3 and~VII.3]{9996629}}]\label{prop:stoch_loc_conv_rate}
    Consider the stochastic localization process Eq.~\eqref{eq:stoch_loc_process}. For any $t>0$,
    \begin{equation}
        \operatorname{W}_2\left(p,\mathcal{L}\left(\bm{m}\left(\bm{u}_t\right)\right)\right)\leq\sqrt{\frac{N}{t}}.
    \end{equation}
\end{prop}
Later, in Section~\ref{sec:stoch_loc_robust}, we will consider a randomized rounding procedure for $\bm{m}$ to force it onto the hypercube which will allow us to upgrade our Wasserstein distance guarantees into total variation distance guarantees. For now, note that convergence in the normalized Wasserstein distance $\frac{1}{\sqrt{N}}\operatorname{W}_2\left(p,q\right)$ implies convergence in distribution with respect to the normalized Euclidean metric $\frac{1}{\sqrt{N}}\left\lVert\cdot\right\rVert_2$.

Before proving this main result, we prove some supplemental results. In what follows, we let
\begin{equation}
    \bm{\varSigma}_{\bm{y}}=\sum_{\bm{z}\in\left\{-1,1\right\}^N}\left(\bm{z}-\bm{m}_{\bm{y}}\right)\otimes\left(\bm{z}-\bm{m}_{\bm{y}}\right)^\intercal p_{\bm{y}}\left(\bm{z}\right)
\end{equation}
denote the covariance matrix of $p_{\bm{y}}\left(\bm{z}\right)$.
\begin{lemma}[Differential equation governing the expected covariance]\label{lem:diff_exp_cov}
    Let $\bm{u}_t$ be governed by the stochastic localization process. The following relations hold:
    \begin{align}
        \dd{\bm{m}_{\bm{u}_t}}&=\bm{\varSigma}_{\bm{u}_t}\dd{\bm{B}_t};\label{eq:m_diff_eq}\\
        \dv{t}\mathbb{E}\left[\bm{\varSigma}_{\bm{u}_t}\right]&\preceq -\mathbb{E}\left[\bm{\varSigma}_{\bm{u}_t}\right]^2.\label{eq:cov_diff_ineq}
    \end{align}
\end{lemma}
\begin{proof}
    Eq.~\eqref{eq:m_diff_eq} was shown in \cite[Eq.~(11)]{eldan2022109392}; explicitly, by definition of $\bm{m}_{\bm{u}_t}$,
    \begin{equation}
        \begin{aligned}
            \dd{\bm{m}_{\bm{u}_t}}&=\sum_{\bm{z}\in\left\{-1,1\right\}^N}\bm{z}\dd{p_{\bm{u}_t}\left(\bm{z}\right)}\\
            &=\sum_{\bm{z}\in\left\{-1,1\right\}^N}\left(\bm{z}\otimes\left(\bm{z}-\sum_{\bm{w}\in\left\{-1,1\right\}^N}\bm{w}p_{\bm{u}_t}\left(\bm{w}\right)\right)^\intercal\right)p_{\bm{u}_t}\left(\bm{z}\right)\dd{\bm{B}_t}\\
            &=\left(\sum_{\bm{z}\in\left\{-1,1\right\}^N}\bm{z}\otimes\bm{z}^\intercal p_{\bm{u}_t}\left(\bm{z}\right)-\bm{m}_{\bm{u}_t}\otimes\bm{m}_{\bm{u}_t}^\intercal\right)\dd{\bm{B}_t}\\
            &=\bm{\varSigma}_{\bm{u}_t}\dd{\bm{B}_t}.
        \end{aligned}
    \end{equation}

    We now show Eq.~\eqref{eq:cov_diff_ineq}, following a similar strategy to \cite[Eq.~(11)]{eldan2020taming}. Integrating Eq.~\eqref{eq:m_diff_eq} gives:
    \begin{equation}\label{eq:m_squared_rel}
        \begin{aligned}
            \bm{m}_{\bm{u}_t}-\bm{m}_{\bm{u}_0}&=\int_{s=0}^t\bm{\varSigma}_{\bm{u}_s}\dd{\bm{B}_s}\\
            \implies\mathbb{E}_{\bm{B}}\left[\left(\bm{m}_{\bm{u}_t}-\bm{m}_{\bm{u}_0}\right)\otimes\left(\bm{m}_{\bm{u}_t}-\bm{m}_{\bm{u}_0}\right)^\intercal\right]&=\mathbb{E}_{\bm{B}}\left[\left(\int_{s=0}^t\bm{\varSigma}_{\bm{u}_s}\dd{\bm{B}_s}\right)\otimes\left(\int_{s=0}^t\bm{\varSigma}_{\bm{u}_s}\dd{\bm{B}_s}\right)^\intercal\right].
        \end{aligned}
    \end{equation}
    Now, as $p_{\bm{u}_t}\left(A\right)$ is a martingale for all $A\subseteq\left\{-1,1\right\}^N$ by Proposition~\ref{prop:stoch_loc},
    \begin{equation}
        \dv{t}\mathbb{E}_{\bm{B}}\left[\sum_{\bm{z}\in\left\{-1,1\right\}^N}\bm{z}\otimes\bm{z}^\intercal p_{\bm{u}_t}\left(\bm{z}\right)\right]=\bm{0}
    \end{equation}
    and
    \begin{equation}
        \dv{t}\mathbb{E}_{\bm{B}}\left[\bm{m}_{\bm{u}_t}\right]=\bm{0}.
    \end{equation}
    In particular,
    \begin{equation}\label{eq:m_cov_rel}
        \begin{aligned}
            -\dv{t}\mathbb{E}_{\bm{B}}\left[\left(\bm{m}_{\bm{u}_t}-\bm{m}_{\bm{u}_0}\right)\otimes\left(\bm{m}_{\bm{u}_t}-\bm{m}_{\bm{u}_0}\right)^\intercal\right]&=-\dv{t}\mathbb{E}_{\bm{B}}\left[\bm{m}_{\bm{u}_t}\otimes\bm{m}_{\bm{u}_t}^\intercal\right]\\
            &=\dv{t}\mathbb{E}_{\bm{B}}\left[\sum_{\bm{z}\in\left\{-1,1\right\}}\left(\bm{z}\otimes\bm{z}^\intercal-\bm{m}_{\bm{u}_t}\otimes\bm{m}_{\bm{u}_t}^\intercal\right)p_{\bm{u}_t}\left(\bm{z}\right)\right]\\
            &=\dv{t}\mathbb{E}_{\bm{B}}\left[\bm{\varSigma}_{\bm{u}_t}\right].
        \end{aligned}
    \end{equation}
    Furthermore, by It{\^o}'s isometry~\cite[Corollary~3.1.7]{oksendal2003} followed by Fubini's theorem,
    \begin{equation}\label{eq:ito_iso}
        \mathbb{E}_{\bm{B}}\left[\left(\int_{s=0}^t\bm{\varSigma}_{\bm{u}_s}\dd{\bm{B}_s}\right)\otimes\left(\int_{s=0}^t\bm{\varSigma}_{\bm{u}_s}\dd{\bm{B}_s}\right)^\intercal\right]=\mathbb{E}_{\bm{B}}\left[\int_0^t\bm{\varSigma}_{\bm{u}_s}^2\dd{s}\right]=\int_0^t\mathbb{E}_{\bm{B}}\left[\bm{\varSigma}_{\bm{u}_s}^2\right]\dd{s}.
    \end{equation}
    Taking the time derivative of Eq.~\eqref{eq:m_squared_rel} and applying Eqs.~\eqref{eq:m_cov_rel} and~\eqref{eq:ito_iso} thus gives:
    \begin{equation}
        \dv{t}\mathbb{E}_{\bm{B}}\left[\bm{\varSigma}_{\bm{u}_t}\right]=-\mathbb{E}_{\bm{B}}\left[\bm{\varSigma}_{\bm{u}_s}^2\right].
    \end{equation}
    Eq.~\eqref{eq:cov_diff_ineq} then follows by noting $\mathbb{E}_{\bm{B}}\left[\bm{\varSigma}_{\bm{u}_s}^2\right]\succeq\mathbb{E}_{\bm{B}}\left[\bm{\varSigma}_{\bm{u}_s}\right]^2$.
\end{proof}
We can use Lemma~\ref{lem:diff_exp_cov} to compute an explicit bound on the operator norm of $\bm{\varSigma}_t$ as a function of $t$.
\begin{lemma}[Bound on the expected covariance]\label{lem:exp_cov_bound}
    Let $\bm{u}_t$ be governed by the stochastic localization process. For all $t\in\mathbb{R}^+$,
    \begin{equation}\label{eq:cov_spec_bound}
        \mathbb{E}_{\bm{B}}\left[\bm{\varSigma}_{\bm{u}_t}\right]\preceq\frac{1}{t}\mathbf{I}_N.
    \end{equation}
\end{lemma}
\begin{proof}
    To simplify notation we define:
    \begin{equation}
        \overline{\bm{\varSigma}}_t:=\mathbb{E}_{\bm{B}}\left[\bm{\varSigma}_{\bm{u}_t}\right]\succeq\bm{0}.
    \end{equation}
    Note that any $0$ eigenvectors $\bm{v}_i$ of $\overline{\bm{\varSigma}}_t$ trivially satisfy:
    \begin{equation}
        \bm{v}_i^\intercal\overline{\bm{\varSigma}}_t\bm{v}_i=0\leq\frac{1}{t}\left\lVert\bm{v}_i\right\rVert_2^2,
    \end{equation}
    so we can assume $\overline{\bm{\varSigma}}_t$ is invertible (i.e., positive definite) WLOG by taking $\overline{\bm{\varSigma}}_t\to\overline{\bm{\varSigma}}_t+\sum_i\bm{v}_i\otimes\bm{v}_i^\intercal$ and noting this perturbation does not affect the nonzero eigenvectors.

    Now, we rewrite Eq.~\eqref{eq:cov_diff_ineq} from Lemma~\ref{lem:diff_exp_cov}:
    \begin{equation}
        \begin{aligned}
            \dv{t}\overline{\bm{\varSigma}}_t&\preceq-\overline{\bm{\varSigma}}_t^2\\
            \iff-\overline{\bm{\varSigma}}_t^{-1}\left(\dv{t}\overline{\bm{\varSigma}}_t\right)\overline{\bm{\varSigma}}_t^{-1}&\succeq\mathbf{I}_N\\
            \iff\dv{t}\left(\overline{\bm{\varSigma}}_t^{-1}\right)&\succeq\mathbf{I}_N.
        \end{aligned}
    \end{equation}
    Integrating over any interval $\left[0,t\right]$ with $t>0$ and noting that $\overline{\bm{\varSigma}}_0^{-1}$ is positive semidefinite then yields:
    \begin{equation}
        \begin{aligned}
            \overline{\bm{\varSigma}}_t^{-1}-\overline{\bm{\varSigma}}_0^{-1}&\succeq t\mathbf{I}_N\\
            \implies\overline{\bm{\varSigma}}_t^{-1}&\succeq t\mathbf{I}_N\\
            \implies\overline{\bm{\varSigma}}_t&\preceq\frac{1}{t}\mathbf{I}_N.
        \end{aligned}
    \end{equation}
\end{proof}

We now have the tools to prove Proposition~\ref{prop:stoch_loc_conv_rate}.
\begin{proof}[Proof of Proposition~\ref{prop:stoch_loc_conv_rate}]
    Using Lemma~\ref{lem:exp_cov_bound} and taking the trace of both sides of Eq.~\eqref{eq:cov_spec_bound},
    \begin{equation}
        \mathbb{E}_{\bm{B}}\left[\sum_{\bm{z}\in\left\{-1,1\right\}^N}\left\lVert\bm{z}-\bm{m}_{\bm{u}_t}\right\rVert_2^2 p_{\bm{u}_t}\left(\bm{z}\right)\right]\leq\frac{N}{t}.
    \end{equation}
    By definition, $\sum_{\bm{z}\in\left\{-1,1\right\}^N}\left\lVert\bm{z}-\bm{m}_{\bm{u}_t}\right\rVert_2^2 p_{\bm{u}_t}\left(\bm{z}\right)$ upper bounds the Wasserstein distance of order $2$ between $p_{\bm{u}_t}$ and the Dirac delta function $\delta_{\bm{m}_{\bm{u}_t}}$ centered at $\bm{m}_{\bm{u}_t}$. That is,
    \begin{equation}
        \mathbb{E}_{\bm{B}}\left[\operatorname{W}_2\left(p_{\bm{u}_t},\delta_{\bm{m}_{\bm{u}_t}}\right)^2\right]\leq\frac{N}{t}.
    \end{equation}
    As $\operatorname{W}_2\left(\cdot,\cdot\right)^2$ is jointly convex in its arguments, we have by Jensen's inequality:
    \begin{equation}
        \operatorname{W}_2\left(\mathbb{E}_{\bm{B}}\left[p_{\bm{u}_t}\right],\mathbb{E}_{\bm{B}}\left[\delta_{\bm{m}_{\bm{u}_t}}\right]\right)^2\leq\frac{N}{t}.
    \end{equation}
    The final result then follows by noting that $\mathbb{E}_{\bm{B}}\left[p_{\bm{u}_t}\right]=p$ and $\mathbb{E}_{\bm{B}}\left[\delta_{\bm{m}_{\bm{u}_t}}\right]=\bm{m}_{\bm{u}_t}$.
\end{proof}

\subsection{Stochastic localization is robust to discretization and tilted-mean approximation}\label{sec:stoch_loc_robust}

 The prescription of Stochastic Localization (SL) outlined in the previous section requires us to run the exact process below for long enough time $T$ (set according to Proposition~\ref{prop:stoch_loc_conv_rate}):
\begin{equation}
    \bm{u}_T = \int_0^T \bm{m}(\bm{u}_t)dt + \bm{W}_T,
\end{equation}
where $\bm{W}_T$ denotes Brownian motion. 
Let us call the law of paths $\bm{u}:=(\bm{u}_t)_{t\in[0,T]}$ generated by the exact SL process above $P$. Observe that the outcome space of this law is the space of continuous functions $C([0,T], \R^N)$.

In practice, we are only able to run a numerical approximation of the exact SL process and this is different in two important ways. First, the process must be discretized into $K$ stepsof size $\delta = T/L$. We denote these discrete grid points $\left\{t_k=k\delta\right\}_{k=0}^K$. Secondly, in practice we have access only to an estimator $\hat{\bm{m}} (\cdot)$ of the magnetizations $\bm{m} (\cdot)$ as a function of tilt $\bm{u}_t$. As a result, the numerical SL process is as follows
\begin{equation}
   \hat{\bm{u}}_{k+1} = \hat{\bm{u}}_{k} + \hat{\bm{m}}(\hat{\bm{u}}_{k})\delta + \bm{W}_{t_{k+1}} - \bm{W}_{t_{k}}.
\end{equation}
Because this stochastic process is discrete, the law over paths $(\hat{\bm{u}}_{k})_{k=0}^{K}$ that it generates is over a different outcome space than that of the exact SL process above. For our proofs, it will also be important to consider an ``interpolated'' version of the discretized numerical process which shares a sample space with the exact SL process by ensuring that both their laws are over $C([0,T], \R^N)$. As such, we interpolate between the grid points $(\hat{\bm{u}}_k)_{i=0}^{K}$ in the following specific way---which we will motivate later---to make the paths continuous across $[0,T]$. Define $\eta$ as the function tracking the most recent grid point, i.e., $\eta(t) = t_k$ for $t_k\leq t<t_{k+1}$. We define the interpolated numerical process $(\hat{\bm{u}}_t)_{t\in[0,T]}$ as that obtained by continuing the Brownian motion but freezing the drift to that at $\eta(t)$. That is,
\begin{equation}
\hat{\bm{u}}_\tau =    \int_0^\tau \hat{\bm{m}}(\hat{\bm{u}}_{\eta(t)})dt + \bm{W}_\tau
\end{equation}
for all $\tau\in\left[0,T\right]$. Notice that the set of tilts from this process evaluated at the grid times still follows the same law as the paths generated by the discrete numerical process above, that is, $(\hat{\bm u}_{t_k})_{k=0}^K\overset{d}{=} (\hat{\bm{u}}_{k})_{k=0}^{K}$. 
In analogy to exact SL, let $Q$ be the law over (now continuous) paths $\hat{\bm{u}}:=(\hat{\bm{u}}_t)_{t\in[0,T]}$ generated by this interpolated numerical SL process. 

Taking these numerical errors into account and using the convergence rate of the exact stochastic localization process (Proposition~\ref{prop:stoch_loc_conv_rate}) we now prove our main result on the convergence of algorithmic stochastic localization.
\begin{theorem}[General convergence of algorithmic stochastic localization]\label{thm:num_guar_sl}
Let $p$ be a distribution on $\{\pm1\}^N$, and let $\bm{m}(\bm u)$ denote the mean of its tilted distribution at tilt $\bm u\in\mathbb R^N$. Suppose $\bm{m}$ is continuous. Let $\hat{\bm{m}}:\mathbb R^N\to[-1,1]^N$ be a measurable estimator satisfying the accuracy bound
\begin{align}
    \frac{1}{\sqrt{N}}\sup_{\bm u\in\mathbb R^N}\norm{\hat{\bm{m}}(\bm u)-\bm{m}(\bm u)}_2\leq \rho.
\end{align}
For a grid $\left\{t_k=k\delta\right\}_{k=0}^K$, define the algorithmic stochastic localization process by
\begin{align}
    \hat{\bm u}_{k+1}
    =\hat{\bm u}_k+\delta\,\hat{\bm{m}}(\hat{\bm u}_k)+\bm W_{t_{k+1}}-\bm W_{t_k},
    \qquad \hat{\bm u}_0=0,
\end{align}
and let $R:[-1,1]^N\to\{\pm1\}^N$ denote coordinate-wise randomized rounding. Then, for every $\epsilon>0$, there exist
\begin{align}
    T=\poly(N,1/\epsilon),\qquad
    \delta^{-1}=\poly(N,1/\epsilon),\qquad
    \rho^{-1}=\poly(N,1/\epsilon)
\end{align}
such that, with $K=T/\delta=\poly(N,1/\epsilon)$,
\begin{align}
    \TVD\!\left(
        p,\,
        \mathcal L\!\left(R\!\left(\hat{\bm{m}}(\hat{\bm u}_K)\right)\right)
    \right)
    \leq \epsilon.
\end{align}
Consequently, if $\hat{\bm{m}}\left(\cdot\right)$ can be evaluated in polynomial time to inverse-polynomial accuracy, then stochastic localization yields a polynomial-time sampler whose output distribution is polynomially close to $p$ in total variation distance.
\end{theorem}

\begin{proof}
In the following, we omit boldface notation for vectors to improve legibility. Nevertheless, do note that the mean magnetizations, tilts, and Brownian motion are all vector-valued quantities. We begin by showing convergence in Wasserstein distance and, at the end, show that we can choose values for $T,\delta,\rho$ such that we also have convergence in total variation distance.

By the triangle inequality, 
\begin{align}\label{eq:w2_triangle_ineq}
    W_2(p, \mathcal{L}(R(\hat{m}(\hat{u}_T))))
    &\leq \underbrace{W_2(p, \mathcal{L}(R(\hat{m}(u_T))))}_{\mathrm{(I)}}
    + \underbrace{W_2(\mathcal{L}(R(\hat{m}(u_T))), \mathcal{L}(R(\hat{m}(\hat{u}_T))))}_{\mathrm{(II)}}
\end{align}
where we have introduced a sample \textit{exact} SL path $(u_t)_{t=0}^{T} \sim P$. In words, $\mathcal{L}(R(\hat{m}(u_T)))$ denotes the law on $\{\pm 1\}^N$ induced by applying the rounded estimator $R(\hat{m}(\cdot))$ to the terminal tilt $u_T$ of a sample exact SL path. 

Let us first bound the second term $\mathrm{(II)}$, which measures how different this distribution is from the one obtained by applying the same map to the terminal tilt of a sample \textit{numerical} SL path $(\hat{u}_t)_{t\in[0,T]}\sim Q$ instead:
\begin{align*}
    W_2( \mathcal{L}(R(\hat{m}(u_T)), \mathcal{L}(R(\hat{m}(\hat{u}_T))).
\end{align*}
First, observe that the two laws being compared are supported on the bounded space $\{\pm 1\}^N$ in which any two points $z_1$ and $z_2$ are at most distance $\norm{z_1 - z_2}_2 = 2\sqrt{N}$ away. On  a bounded outcome space, the Wasserstein distance can be bounded in terms of total variation distance using the standard argument:
\begin{align*}
W_2(\mu, \nu)^2 &= \inf_{\pi \in \Pi (\mu,\nu)}\E_{(x,y)\sim \pi}{\norm{X-Y}_2^2}\\
&\leq  \E_{(x,y)\sim \pi}{\norm{X-Y}_2^2} \quad \text{using the standard TVD coupling $\pi$}\\
&= \PP(X \neq Y) \norm{X-Y}_2^2\ \\
&\leq 4N\TVD(\mu,\nu).
\end{align*}
We may now pass from $W_2$ to $\TVD$:
\begin{align*}
    W_2( \mathcal{L}(R(\hat{m}(u_T)), \mathcal{L}(R(\hat{m}(\hat{u}_T)))  \leq \sqrt{4N\TVD(\mathcal{L}(R(\hat{m}(u_T)), \mathcal{L}(R(\hat{m}(\hat{u}_T)))}.
\end{align*}
Next, we use the data-processing inequality to shift the $\TVD$ above to that between the laws of the paths themselves. Notice that the laws over $\{\pm 1\}^N$ are induced by a nested composition of functions which ultimately act on the paths $u\sim P$ and $\hat{u}\sim Q$. More explicitly,
\begin{align*}
\TVD(\mathcal{L}(R(\hat{m}(u_T)), \mathcal{L}(R(\hat{m}(\hat{u}_T)))=\TVD(\mathcal{L}(R(\hat{m}(\operatorname{eval}_T (u))), \mathcal{L}(R(\hat{m}(\operatorname{eval}_T(\hat{u})) ))),
\end{align*}
where $\operatorname{eval}_T(u)$ projects the path $u:\left[0,T\right]\to\mathbb{R}^N$ to its value at $T$. To invoke the data-processing inequality, we must ensure that each mapping in the composition on the right-hand side above is measurable.

We begin with the outermost function first, which is the randomized rounding step $R(\cdot)$. An equivalent formulation of this operation is the deterministic function $R_\text{det}$ defined on two variables:
\begin{align*}
    R_{\mathrm{det}} : [-1,1] \times [0,1] &\to \{\pm 1\}, \\
    (m,v) &\mapsto
    \begin{cases}
        +1, & \text{if } v < \dfrac{1+m}{2}, \\[6pt]
        -1, & \text{if } v \ge \dfrac{1+m}{2}.
    \end{cases}
\end{align*}
Then, taking $v \sim \lambda :=\operatorname{Unif}[0,1]$ as a uniformly random seed recovers the same randomized rounding as $R(\cdot)$. That is, $\mathcal{L}(R(u)) := \mathcal{L}(R_\mathrm{det}(u,v))$.
\begin{prop}
    $R_{\mathrm{det}} : [-1,1] \times [0,1] \to \{+1,-1\}$ is measurable.
\end{prop}
\begin{proof}
We prove the claim that $R_\text{det}$ is measurable by showing that the preimage of an event in $\{+1,-1\}$ is an event in $[-1,1] \times [0,1]$ \textit{i.e.} a Borel set. Consider the event $\{+1\}$. Then its preimage is
\begin{align*}
    R_\text{det}^{-1}(\{+1\}) &= \{(m,v)\in  [-1,1] \times [0,1] :2v-m <1\}\\
    &= g^{-1}((-\infty, 1))
\end{align*}
where we have denoted the constraint by $g(m,v) := 2v-m$. Since $g$ is explicitly continuous in both its arguments, $g^{-1}((-\infty, 1))$ is an open set and consequently a Borel set in $[-1,1] \times [0,1]$. As for the event $\{-1\}$, it is the case that
\begin{align*}
    R_\text{det}^{-1}(\{-1\}) &= (g^{-1}((-\infty, 1)))^c
\end{align*}
and the complement of a Borel set is also a Borel set. We have thus shown that $R_\text{det}$ is a measurable function. 
\end{proof}
It follows from the preceding claim and the data-processing inequality that (recalling $\lambda :=\operatorname{Unif}[0,1]$ is the law of $v$):
\begin{align*}
    \TVD(\mathcal{L}(R_\mathrm{det}(\hat{m}(\operatorname{eval}_T(\hat{u})),v)), \mathcal{L}(R_\mathrm{det}(\hat{m}(\operatorname{eval}_T(u)),v))) &\leq \TVD(\mathcal{L}(\hat{m}(\operatorname{eval}_T(\hat{u}))) \otimes \lambda )), \mathcal{L}(\hat{m}(\operatorname{eval}_T(u))) \otimes \lambda ))) \\
    &\leq \TVD(\mathcal{L}(\hat{m}(\operatorname{eval}_T(\hat{u}))) , \mathcal{L}(\hat{m}(\operatorname{eval}_T(u))))\\
    &\leq\TVD(\mathcal{L}(\operatorname{eval}_T(\hat{u})) , \mathcal{L}(\operatorname{eval}_T(u))),
\end{align*}
with the final line following by the assumed measurability of $\hat{m}$.

Now, we note that $\operatorname{eval}_T$ is continuous, and therefore measurable, when we equip the space of continuous paths $C([0,T], \R^N)$ with metric given by the supremum distance: 
\begin{equation}
    d_\infty (u_1, u_2) := \sup_{t\in [0,T]}\norm{u_1(t) - u_2(t)}_2.
\end{equation}
\begin{prop}
    $\operatorname{eval}_T : C([0,T], \R^N) \to \R^N $ is measurable.
\end{prop}
\begin{proof}
    Let $u_1, u_2 \in C([0,T], \R^N)$, then
    \begin{align*}
         \norm{\operatorname{eval}_T(u_1) - \operatorname{eval}_T(u_2)}_2 &= \norm{u_1(T) - u_2(T)}_2 \\
         &\leq d_\infty(u_1, u_2)
    \end{align*}
    Thus, $\operatorname{eval}_T$ is $1$-Lipschitz $\implies$ continuous $\implies$ measurable. 
\end{proof}
The data-processing inequality therefore finally gives:
\begin{align*}
    \TVD( \mathcal{L}(R(\hat{m}(u_T)), \mathcal{L}(R(\hat{m}(\hat{u}_T)))  \leq \TVD(\mathcal{L}(u),  \mathcal{L}(\hat{u})) = \TVD(P, Q).
\end{align*} 
Now, we use Pinsker's inequality to pass from $\TVD$ to the Kullback–Leibler (KL) divergence:
\begin{align*}
    \TVD(P,Q) &\leq \sqrt{\frac{1}{2}\KL(P\|Q)} = \sqrt{\frac{1}{2}\E_P\left[\ln\frac{dP}{dQ}\right]} = \sqrt{-\frac{1}{2}\E_P\left[\ln\frac{dQ}{dP}\right]},
\end{align*}
where the final equality holds when $\ln\frac{dQ}{dP}$ is well-defined~\cite[Theorem~5]{bermudez2025proofsfolkloretheoremsradonnikodym}, which we will later see to be the case.

Let us focus on the Radon-Nikodym derivative $\frac{dQ}{dP}$. Operationally, this represents how the law $P$ must be reweighted in order to obtain the law $Q$. In other words, for any event $A \in \mathcal{F}$
\begin{equation}
    Q(A)  = \int_A\frac{dQ}{dP}(\omega) P(d\omega)
\end{equation}
To evaluate this derivative, it helps to look at the laws $P$ and $Q$ in a different way. Let $u \in C([0,T], \R^N)$. By definition, 
\begin{align*}
    W_t &= u_t - \int_0^tm(u_t)ds && \text{and under $P$, $W_t$ is Brownian};\\
    \tilde{W}_t &= u_t - \int_0^t\hat{m}(u_{\eta(s)})ds && \text{and under $Q$, $\tilde{W}_t$ is Brownian}.
\end{align*}
Then, $W_t$ and $\tilde{W}_t$ are related to each other by
\begin{equation}
    \tilde{W}_t(u) = W_t(u) + \int_0^t\left(\underbrace{m(u_s) -m(u_{\eta(s)})}_{:=\Delta_s(u)}\right)ds
\end{equation}
We have reformulated the question of evaluating $\frac{dQ}{dP}$ to the following: what change of measure takes a law $P$ under which $\tilde{W}_t(u)$ is Brownian motion plus a drift term $\Delta_s(u)$, to the law $Q$, under which $\tilde{W}_t(u)$ becomes pure Brownian motion? This change of measure is exactly what is supplied by Girsanov's Theorem.
\begin{theorem}[Girsanov's Theorem with Novikov's Condition~{\cite[Corollaries~1.1 and~1.2 with $W_t\to -W_t$]{Kazamaki1994}}]\label{thm:girsanovs}
Let $W_t$ denote Brownian motion under the law $P$, and let $\theta_s$ be a progressively measurable process such that for every $t>0$:
\begin{equation}
    \int_0^t\left\lVert\theta_s\right\rVert_2^2\dd{s}<\infty
\end{equation}
almost surely. Assume further that Novikov's Condition is satisfied, i.e., that:
\begin{equation}
    \E\exp\left(\frac{1}{2}\int_0^t\left\lVert\theta_s\right\rVert_2^2\dd{s}\right)<\infty.
\end{equation}
Then the process:
\begin{equation}
    \tilde{W}_t(u) = \int_0^t\theta(s)ds+W_t(u)
\end{equation}
is Brownian motion under the law:
\begin{equation}
    dQ=Z_t dP,
\end{equation}
where:
\begin{equation}
    Z_t := \exp\left(-\int_0^t \theta_s \cdot dW_s - \frac{1}{2}\int_0^t\norm{\theta_s}_2^2 ds\right ).
\end{equation}
\end{theorem}
Assuming the hypotheses hold for our case, an application of Girsanov's Theorem with $\theta_s = \Delta_s$ tells us that $Z_T = \frac{dQ}{dP}$ is precisely the Radon-Nikodym derivative we are interested in.

We now show that the conditions of Theorem~\ref{thm:girsanovs} are satisfied with $\theta_s=\Delta_s$.
\begin{prop}
    $\Delta_s= m(u(s)) - \hat{m}(u(\eta(s))$ satisfies the conditions of Theorem~\ref{thm:girsanovs}.
\end{prop}
\begin{proof}
    Square-integrability and Novikov's Condition are trivially satisfied as $m(u(s)),\hat{m}(u(\eta(s))\in\left[-1,1\right]^N$. This leaves demonstrating $\Delta_s$ is a progressively measurable process, which is implied by $\Delta_s$ being a right-continuous adapted process. Adaptivity follows by the definition of the stochastic localization process and the fact that $\eta(s)\leq s$. Right-continuity follows by the continuity of $m$ and $u$, as well as $\eta$ (and thus $\hat{m}(u(\eta(s)))$) being piecewise constant on intervals $\left[t_k,t_{k+1}\right)$.
\end{proof}

Having satisfied the conditions of Girsanov's Theorem, we may now identify $Z_t$ with the likelihood ratio and obtain the KL divergence as
\begin{align*}
    \KL(P || Q) &= \E_P\left[\int_0^t \Delta_s \cdot dW_s +\frac{1}{2}\int_0^t\norm{\Delta_s}_2^2 ds\right]\\
    &= \E_P\left[\int_0^t \Delta_s \cdot dW_s \right] +  \E_P\left[\frac{1}{2}\int_0^t\norm{\Delta_s}_2^2 ds\right].
\end{align*}
Recall that under $P$, $W_s$ is a Brownian process; by the martingale property of Itô integrals, it thus vanishes in expectation under $P$, giving:
\begin{align*}
    \KL(P || Q) &=\frac{1}{2}\int_0^t\E_P\left[\norm{\Delta_s}_2^2 \right]ds.
\end{align*}
Next, we use the triangle inequality to expand the integrand as
\begin{equation}\label{eq:exp_integrand}
\begin{aligned}
    \E_P\left[\norm{m(u_s)-\hat{m}(u_{\eta(s)})}_2^2\right]
&=
\E_P\left[\norm{m(u_s)-m(u_{\eta(s)})}_2^2\right]
+
\E_P\left[\norm{m(u_{\eta(s)})-\hat{m}(u_{\eta(s)})}_2^2\right]\\
&\quad+
2\E_P\left[\langle m(u_{\eta(s)})-\hat{m}(u_{\eta(s)}),m(u_s)-m(u_{\eta(s)})\rangle\right].
\end{aligned}
\end{equation}
We now bound each term. We begin with the final term. First, recall from Proposition~\ref{prop:stoch_loc} that $m\left(u_t\right)$ is a martingale as it corresponds to the exact stochastic localization process. In particular, for any $k$ and $s\in\left[t_k,t_{k+1}\right)$, as $\eta(s)$ is constant on this interval,
\begin{equation}
    \dv{s}\E_P\left[\langle m(u_{\eta(s)})-\hat{m}(u_{\eta(s)}),m(u_s)-m(u_{\eta(s)})\rangle\right]=0.
\end{equation}
As $s=\eta(s)$ when $s=t_k$,
\begin{equation}
    \E_P\left[\langle m(u_{\eta(s)})-\hat{m}(u_{\eta(s)}),m(u_s)-m(u_{\eta(s)})\rangle\right]=0.
\end{equation}
We now move on to the second term on the right-hand side of Eq.~\eqref{eq:exp_integrand}, which is also simple as by assumption:
\begin{equation}
    \norm{m(u)-\hat{m}(u)}^2_2 \leq N\rho^2.
\end{equation}
This leaves only the first term on the right-hand side of Eq.~\eqref{eq:exp_integrand}. First, we have:
\begin{align*}
    0\leq \E_P\norm{m(u_s)-m(u_{\eta(s)})}_2^2
    &= \E_P\norm{m(u_s)}_2^2+\E_P\norm{m(u_{\eta(s)})}_2^2-2\E_P\langle m(u_s),m(u_{\eta(s)})\rangle\\
    &= \E_P\norm{m(u_s)}_2^2-\E_P\norm{m(u_{\eta(s)})}_2^2,
\end{align*}
where in the final line we used the martingale property once more to simplify
\begin{equation}
    \begin{aligned}
        \E_P\langle m(u_s),m(u_{\eta(s)})\rangle &= \E_P\langle m(u_{\eta(s)}),m(u_{\eta(s)})\rangle\\
        &= \E_P \norm{m(u_{\eta(s)})}_2^2.
    \end{aligned}
\end{equation}
It is easy to see by this same property that $\E_P\norm{m(u_s)}_2^2$ is a non-decreasing function of $s$, i.e., since $[\eta(s) , s] \subseteq [t_k, t_{k+1}]$ for some $k$,
\begin{align*}
    \E_P\norm{m(u_s)-m(u_{\eta(s)})}_2^2
    \leq \E_P\norm{m(u_{t_{k+1}})}_2^2-\E_P\norm{m(u_{t_{k}})}_2^2.
\end{align*}
We can then write the telescoping sum:
\begin{align*}
    \int_0^t\E_P\norm{m(u_s)-m(u_{\eta(s)})}_2^2
    &\leq \delta\sum_{k=0}^{K-1}\left(\E_P\norm{m(u_{t_{k+1}})}_2^2-\E_P\norm{m(u_{t_k})}_2^2\right)\\
    &\leq \delta\left(\E_P\norm{m(u_T)}_2^2-\E_P\norm{m(u_0)}_2^2\right)\\
    &\leq \delta N.
\end{align*}
where in the last line we have used the boundedness of $m(u_t) \in [-1,+1]^N$ implying $\norm{m(u_T)}_2^2 \leq N$.

Putting everything together, it follows that: 
\begin{equation}
    \KL(P || Q) \leq \frac{1}{2}(N\delta  + N T\rho^2),
\end{equation}
and thus returning to Eq.~\eqref{eq:w2_triangle_ineq}:
\begin{equation}
    \mathrm{(II)} \leq \sqrt{2N}(N\delta  + N T\rho^2)^{1/4}.
\end{equation}

We now bound the other term in Eq.~\eqref{eq:w2_triangle_ineq} $\mathrm{(I)}$. By applying the triangle inequality again,
\begin{align*}
   \mathrm{(I)} &\leq  W_2(p, \mathcal{L}(R(m(u_T)))) +  W_2(\mathcal{L}(R(m(u_T))), \mathcal{L}(R(\hat{m}(u_T))))  
\end{align*}
Since $p$ is already supported on $\{\pm 1\}^N$, randomized rounding leaves it invariant, namely that $X\sim p$ implies $\mathcal{L}(R(X))=p$. We now apply the inequality \cite[Lemma~VII.3]{9996629}:
\begin{equation}
    \frac{1}{\sqrt{N}}\operatorname{W}_2\left(r_1,r_2\right)\leq 2\sqrt{\frac{1}{\sqrt{N}}\operatorname{W}_2\left(p_1,p_2\right)},
\end{equation}
where $r_1$ and $r_2$ are the laws obtained from applying the randomized rounding function $R$ to samples $X\sim p_1$ and $\sim p_2$ respectively, to obtain
\begin{align*}
    \mathrm{(I)}&\leq  2\sqrt{\sqrt{N}\,W_2(p, \mathcal{L}(m(u_T)))} +  2\sqrt{\sqrt{N}\,W_2(\mathcal{L}(m(u_T)), \mathcal{L}(\hat{m}(u_T)))} \\
    & \leq 2\sqrt{N}\Big( \frac{1}{T^{1/4}}+ 2\sqrt{\rho}\Big) \\
\end{align*}
The first bound arises from the convergence guarantee of exact SL (Proposition~\ref{prop:stoch_loc_conv_rate}), and the second from the definition of $W_2$. 

Returning to Eq.~\eqref{eq:w2_triangle_ineq}, we get:
\begin{align*}
    \frac{1}{\sqrt{N}}W_2(p, \mathcal{L}(R(\hat{m}(\hat{u}_T))) \leq  \frac{2}{T^{1/4}}+ 2\sqrt{\rho} + \sqrt{2}(N\delta  + N T\rho^2)^{1/4}.
\end{align*}
Furthermore, as both $p$ and $\mathcal{L}(R(\hat{m}(\hat{u}_T))$ are distributions with support on the hypercube,
\begin{equation}\label{eq:final_tvd_w2_bound}
    \TVD\left(p,\mathcal{L}(R(\hat{m}(\hat{u}_T))\right)\leq W_2(p, \mathcal{L}(R(\hat{m}(\hat{u}_T)))^2\leq N\left(\frac{2}{T^{1/4}}+ 2\sqrt{\rho} + \sqrt{2}(N\delta  + N T\rho^2)^{1/4}\right)^2.
\end{equation}
For the following parameterizations,
\begin{equation}\label{eq:choice_of_params_asl}
\boxed{
T = \frac{2^8 N^2}{\epsilon^2},
\qquad
\rho = \frac{\epsilon^2}{2^{10}N^{5/2}},
\qquad
K = \left\lceil \frac{2^{20}N^5}{\epsilon^4} \right\rceil,
\qquad
\delta = \frac{T}{K} \leq \frac{\epsilon^2}{2^{12}N^3}.
}
\end{equation}
this gives:
\begin{align*}
    \TVD\left(p,\mathcal{L}(R(\hat{m}(\hat{u}_T))\right)\leq\left(\frac{\sqrt{\epsilon}}{2}+ \frac{\sqrt{\epsilon}}{16} + \frac{\sqrt{\epsilon}}{2^{9/4}}\right)^2<\epsilon.
\end{align*}
For completeness, to achieve $\epsilon$ error in $W_2$, it is easy to see from Eq.~\eqref{eq:final_tvd_w2_bound} that the following choices of parameters suffice:
\begin{equation}\label{eq:choice_of_params_asl_w2}
\boxed{
T = \frac{2^8 N^2}{\epsilon^4},
\qquad
\rho = \frac{\epsilon^4}{2^{10}N^{5/2}},
\qquad
K = \left\lceil \frac{2^{20}N^5}{\epsilon^8} \right\rceil,
\qquad
\delta = \frac{T}{K} \leq \frac{\epsilon^4}{2^{12}N^3}.
}
\end{equation}
\end{proof}

\section{Barvinok's method and estimating tilted means for the quantum $p$-spin model}\label{sec:quasi_poly_est_alg}

As an application of Theorem~\ref{thm:num_guar_sl}, we consider using it as a means to Gibbs sample from the quantum $p$-spin model~\cite{swingle_bosonic_2023}. We recall the quantum $p$-spin model is an ensemble of local Hamiltonians defined in the following way.
\begin{defn}[Quantum $p$-spin model]
   \begin{equation}
    H = \sum_{I \in \binom{[N]}{q}\times {\{X,Y,Z\}^q}} J_I \hat{\sigma}_I
\end{equation}
 with disorder $J_I \sim_{\text{iid}} \mathcal{N}(0, \frac{(q-1)!}{3^qN^{q-1}}J^2)$ and constant $q \geq3$. The index $I$ is defined so that, for example, if $I = \bigl(\{1,3,5\}, \{X,Y,X\}\bigr).$ then $\hat{\sigma}_I:= \hat\sigma^{1}_{x}\hat\sigma^{3}_{y}\hat\sigma^{5}_{x}$. Denote by $\supp(\hat\sigma_I)$ the support of the corresponding Pauli string. In the earlier example, this corresponds to qubits $\{1,3,5\}$. 
    \label{dfn:qpspin_cs}
\end{defn}
The model in \cref{dfn:qpspin_cs} is normalized to ensure that each spin is subject to an effective field strength that is $\mathcal{O}(1)$ in the total number of spins $N$. This is because each spin occurs in $O\left(\binom{N-1}{q-1}\right) \approx O\left(\frac{N^{q-1}}{(q-1)!}\right) $ terms.

We now recall our task at hand, which is \emph{Gibbs sampling}.
\begin{defn}[Gibbs sampling the quantum $p$-spin model]
    Given $\bm{J}$ and inverse temperature $\beta$, the task of \emph{Gibbs sampling to error $\epsilon$} is to produce samples from a distribution $q$ approximating:
    \begin{equation}
        p\left(\bm{z}\right)=\frac{\bra{\bm{z}}e^{-\beta\bm{H}}\ket{\bm{z}}}{\Tr\left(e^{-\beta\bm{H}}\right)}
    \end{equation}
    to error $\epsilon$ total variation distance. Here, $\ket{\bm{z}}$ denotes a computational basis state, though by the rotational invariance of the quantum $p$-spin model our results hold for any choice of product state basis.
\end{defn}

In order to demonstrate an efficient sampling algorithm for the quantum $p$-spin model, by Theorem~\ref{thm:num_guar_sl} it suffices to consider estimating the mean of the tilted distribution:
\begin{equation}
    p_{\bm{y}}\left(\bm{z}\right)=\frac{\exp\left(\bm{y}^\intercal\cdot\bm{z}\right)\bra{\bm{z}}e^{-\beta\bm{H}}\ket{\bm{z}}}{\sum\limits_{\bm{z}\in\left\{-1,1\right\}^N}\exp\left(\bm{y}^\intercal\cdot\bm{z}\right)\Tr\left(e^{-\beta\bm{H}}\right)}.
\end{equation}
This, in turn, can be viewed as a quantum many-body system described by the effective \emph{tilted free energy}:
\begin{equation}
    F_{\bm{y}}:=\ln Z_{\bm{y}}:=\ln \Tr\left(\exp\left(\sum_{i=1}^N y_i\sigma_i^Z\right)\exp\left(-\beta H\right)\right),
\end{equation}
in the sense that local magnetizations of $F_{\bm{y}}$ correspond to the mean of the tilted distribution $p_{\bm{y}}$. This can be made more rigorous via the following Proposition.
\begin{prop}[Error bounding]\label{claim:mag_from_finite_difference}
    Estimating the tilted free energy $\ln Z_{\bm{y}}\;\forall \bm{y} \in \R^N $ to additive error $\epsilon$ is sufficient for estimating the tilted magnetization $m_i\left(\bm{y}\right) \;\forall i\in [N]$ to additive error $\epsilon' = O(\epsilon)$.
\end{prop}
\begin{proof}
 Using shorthand notation $Z_{\bm{y}}:= \Tr[e^{\sum_{i=1}^{N}  y_i \sigma_z^i}e^{-\beta H}]$, consider perturbing the $i$-th spin field by some constant $t>0$:
\begin{align}
   Z_{\bm{y}+t\bm{e}_i} &=  \Tr[e^{t\sigma_z^i + \sum_{i=1}^{N}  y_i \sigma_z^i}e^{-\beta H}] \notag \\
   &=\Tr[e^{t\sigma_z^i}e^{\sum_{i=1}^{N}  y_i \sigma_z^i}e^{-\beta H}]\notag\\
   &=\Tr[e^{\sum_{i=1}^{N}  y_i \sigma_z^i}e^{-\beta H}] \cosh t +  \Tr[\sigma^i_ze^{\sum_{i=1}^{N}  y_i \sigma_z^i}e^{-\beta H}] \sinh t \notag \\
   &=Z_{\bm{y}}\cosh t  +Z_{\bm{y}} m_i \sinh t \notag \\
   \implies m_i(\bm{y}) &=\frac{e^{\ln Z_{\bm{y}+t\bm{e}_i} - \ln Z_{\bm{y}}}- \cosh t}{\sinh t} \label{eq:magnetization-from-free-energy}
\end{align}
We now prove the estimation claim. Suppose estimates $\widehat{\ln Z_{\bm{y}}}$ and $\widehat{\ln Z_{\bm{y}+t\bm{e}_i}}$ are given, each with additive error at most $\epsilon/2$. For clarity, denote the finite difference term by $L_{\bm{y}}:=\ln Z_{\bm{y}+t\bm{e}_i}-\ln Z_{\bm{y}}$ so that $|\hat{L}_{\bm{y}}-L_{\bm{y}}|\leq \epsilon$. Then,
\begin{align*}
   |\hat{m}_i(\bm{y})-m_i(\bm{y})|
   &=\frac{1}{|\sinh t|}\left|e^{\hat{L}_{\bm{y}}}-e^{L_{\bm{y}}}\right| \\
   &=\frac{e^{L_{\bm{y}}}}{|\sinh t|}\left|e^{\hat{L}_{\bm{y}}-L_{\bm{y}}}-1\right| \\
   &\leq \frac{e^{L_{\bm{y}}}}{|\sinh t|}\left |e^\epsilon-1\right |\\
   &=\frac{\cosh t+m_i(\bm{y})\sinh t}{|\sinh t|}\left |e^\epsilon-1\right| && \text{by Eq.~\eqref{eq:magnetization-from-free-energy}}\\
   &\leq \frac{e^{|t|}}{|\sinh t|} \left| e^\epsilon-1 \right| && \text{since $m_i(\bm{y}) \in [-1,1]$} \\ 
   &\leq O(\epsilon) &&  \text{using $\left| e^\epsilon-1 \right| \leq \epsilon e$ whenever $|\epsilon|< 1$}
\end{align*}
\end{proof}
It is known that if the partition function $Z(\beta)\neq 0$ on a disk $|\beta|\leq C$ in the complex $\beta$ plane, then there exists an $N^{O\left(\log\left(N/\epsilon\right)\right)}$-time classical algorithm for estimating the free energy to additive error $\epsilon$ \cite{v011a013,bencs_zeros_2026,zlokapa2026sykthermalexpectationsclassically,zlokapa2026rigorousquasipolynomialtimeclassicalalgorithm}. It also works for other zero-free geometries, including the strip as described in the introduction; we refer to \cite{zlokapa2026sykthermalexpectationsclassically} for a detailed review, but we reproduce the main result here.
\begin{prop}[Barvinok's interpolation method for computing free energies~{\cite[Appendix~A]{zlokapa2026sykthermalexpectationsclassically}}]\label{prop:barvinoks}
    Let $Z\left(\beta\right)=\exp\left(-\beta H\right)$ be the partition function of a system with Hamiltonian $H=\sum_I J_I P_I$ for Pauli operators $\left\{P_I\right\}$. Assume $Z\left(\beta\right)$ is zero-free on either a disk of radius $\beta^\ast$ or a complex strip $-\beta^\ast\leq\Re\left\{\beta\right\}\leq\beta^\ast,-\delta\leq\Im\left\{\beta\right\}\leq\delta$ for some $N$-independent $\beta^\ast>0$ and $\delta>0$. Then, for any $\epsilon>0$ there exists a classical algorithm for estimating $\ln Z\left(\beta\right)$ to additive error $\epsilon$ in time $N^{\operatorname{O}\left(\log\left(N/\epsilon\right)\right)}$.
\end{prop}

While unfortunately we are unable to rigorously prove the zero-freeness of $Z_{\bm{y}}\left(\beta\right)$ down to the spin glass phase transition, in this section we give a mixture of rigorous and physics-based evidence that this is true. Namely:
\begin{enumerate}
    \item In Proposition~\ref{prop:tilt_exp_zero_free} and Theorem~\ref{thm:zero_freeness}, we show that $\E Z_{\bm{y}}\left(\beta\right)$ is zero-free in a constant-radius disk around the origin.
    \item In Theorem~\ref{thm:zero_freeness}, we show that $Z_{\bm{y}=\bm{0}}\left(\beta\right)$ concentrates, and therefore is zero-free in a constant-radius disk around the origin with high probability.
    \item In Section~\ref{sec:conj_conc} we use physics arguments to claim that $Z_{\bm{y}}\left(\beta\right)$ concentrates when $Z_{\bm{y}=\bm{0}}\left(\beta\right)$ concentrates; if this were true, then our first statement would also upgrade to a w.h.p. statement.
    \item Assuming Fisher zeroes exactly correspond to phase transitions~\cite{10.1145/3357713.3384322}, $Z_{\bm{y}}$ is indeed zero-free down to the spin glass phase transition.
\end{enumerate}

\subsection{Zero-freeness of the partition function}

We first show that the untilted partition function has, w.h.p., a constant-sized zero-free region around the origin on the complex $\beta$ plane; in the course of showing this we will also have the tools required for Proposition~\ref{prop:tilt_exp_zero_free}, where we show the expected tilted partition function also has a zero-free region around the origin.

Our arguments here closely follow those of \cite{zlokapa2026rigorousquasipolynomialtimeclassicalalgorithm}, which demonstrated zero-freeness for the SYK model on a constant-sized disk. For this reason we will refer back to the proof of zero-freeness of \cite{zlokapa2026rigorousquasipolynomialtimeclassicalalgorithm}, giving here only a proof overview of what is done there and only explicitly demonstrating here how the quantum $p$-spin model differs. We will find that the main difference is in the proof of Lemma~\ref{lemma:KP_second_moment} for the quantum $p$-spin model, which we will use to demonstrate concentration of the partition function, and thus devote the most space here to this proof.

\subsubsection{Polymer representations of the first and second moments of the quantum $p$-spin model}
Central to the zero-freeness proof of \cite{zlokapa2026rigorousquasipolynomialtimeclassicalalgorithm} is the expansion of the first and second moments of the partition function as a \emph{polymer series}. To develop this idea, we introduce some preliminary definitions.

\begin{defn}[Support-intersection graphs]\label{def:support_intersection_graph}
Consider a set of Pauli strings $I_1,\dots,I_m\in\{I,X,Y,Z\}^N$.

\begin{enumerate}
   \item \textbf{Support-intersection graph.} The graph $G(\bm{I})$ is defined on the vertex set $[m]$, with an edge between distinct vertices $i,j\in[m]$ if
   \begin{equation}
      \supp(I_i)\cap\supp(I_j)\neq\emptyset.
   \end{equation}
   That is, any pair of vertices are connected if their corresponding Pauli strings have overlapping support on at least one spin.

   \item \textbf{Separated support-intersection graph.} Suppose $m$ is even and the collection is additionally equipped with a pairing $\pi\in\Pi([m])$ and replica labeler $\mu:[m]\to\{1,2\}$. For the tuple $\tau:=(\pi,\mu,\bm{I})$, the graph $G^{\mathrm{sep}}(\tau)$ is obtained by restricting the edges of $G(\bm{I})$ to vertices with matching replica label, \textit{i.e.}
   \begin{equation}
      \mu(i)=\mu(j)
      \qquad\text{and}\qquad
      \supp(I_i)\cap\supp(I_j)\neq\emptyset.
   \end{equation}

   \item \textbf{Mixed support-intersection graph.} For the same tuple $\tau:=(\pi,\mu,\bm{I})$, the graph $G^{\mathrm{mix}}(\tau)$ is obtained from $G^{\mathrm{sep}}(\tau)$ by additionally drawing an edge between paired vertices belonging to distinct replicas, \textit{i.e.} whenever
   \begin{equation}
      \{i,j\}\in\pi
      \qquad\text{and}\qquad
      \mu(i)\neq\mu(j).
   \end{equation}
\end{enumerate}
\end{defn}

The polymer sets for both the first- and second-moments will end up consisting of specific sets of Pauli strings whose support intersection graph is connected. 

\begin{defn}[First- and second-moment Polymer sets]\label{def:moment_polymer_sets}
Fix an even integer $m$ and denote by $\Pi([m])$ the set of all possible pairings of $[m]$. A choice of $m$ $q$-local Pauli strings $I_1,\dots,I_m\in\binom{[N]}{q}\times\{X,Y,Z\}^q$ and a pairing $\pi\in\Pi([m])$ are said to be \textit{consistent} if $I_a=I_b$ for all $\{a,b\}\in\pi$. Collect all such consistent pairs into the set
\begin{align*}
   \mathcal{S}_m:=\{(\pi,\bm{I}): \text{$\bm{I}$ is consistent under the pairing $\pi$}\}.
\end{align*}
If the Pauli strings are additionally equipped with a replica labeler $\mu:[m]\to\{1,2\}$, define similarly
\begin{align*}
   \mathcal{T}_m:=\{(\pi,\mu,\bm{I}): \text{$\bm{I}$ is consistent under the pairing $\pi$}\}.
\end{align*}

\begin{enumerate}
   \item \textbf{First-moment polymer set.} For an element $\tau:=(\pi,\bm{I})\in\mathcal{S}_m$, define $G(\tau):=G(\bm{I})$ and let $\mathcal{S}_m^{\text{conn}}$ denote the set of all tuples whose support-intersection graph is connected. The \textit{first-moment polymer set} is defined by
   \begin{equation}
      \mathcal{P}_{\E[Z]}:=\bigcup_{k\geq1}\mathcal{S}_{2k}^{\text{conn}}.
   \end{equation}
   For $\tau=(\pi,\bm{I})\in\mathcal{P}_{\E[Z]}$, define
   \begin{equation}
      \supp(\tau):=\bigcup_{j}\supp(I_j).
   \end{equation}
   The polymer set is equipped with the compatibility relation $\sim$ where $\tau_1,\tau_2\in\mathcal{P}_{\E[Z]}$ satisfy
   \begin{equation}
      \tau_1\sim\tau_2
      \iff
      \supp(\tau_1)\cap\supp(\tau_2)=\emptyset.
   \end{equation}

   \item \textbf{Second-moment polymer set.} For an element $\tau:=(\pi,\mu,\bm{I})\in\mathcal{T}_m$, let $\mathcal{T}_m^{\text{conn}}$ denote the set of all tuples whose mixed support-intersection graph $G^{\mathrm{mix}}(\tau)$ is connected. The \textit{second-moment polymer set} is defined by
   \begin{equation}
      \mathcal{P}_{\E[Z^2]}:=\bigcup_{k\geq1}\mathcal{T}_{2k}^{\text{conn}}.
   \end{equation}
   For $\tau\in\mathcal{P}_{\E[Z^2]}$, define its support within replica $a\in\{1,2\}$ by
   \begin{equation}
      \supp_a(\tau):=\bigcup_{j :  \mu(j)=a}\supp(I_j).
   \end{equation}
   The polymer set is equipped with the compatibility relation $\sim$ where $\tau_1,\tau_2\in\mathcal{P}_{\E[Z^2]}$ satisfy
   \begin{equation}
      \tau_1\sim\tau_2
      \iff
      \supp_a(\tau_1)\cap\supp_a(\tau_2)=\emptyset
      \quad\text{for each }a\in\{1,2\}.
   \end{equation}
\end{enumerate}
\end{defn}

The upshot of defining these polymer sets is that the first and second moments can be exactly expressed in terms of them.

\begin{prop}[The first moment as a polymer expansion]
The first moment of the quantum $p$-spin model can be expanded as a polymer series in terms of the first-moment polymer set $\mathcal{P}_{\E[Z]}$ of \cref{def:moment_polymer_sets}. That is,
\begin{equation}
    \E[\hat{Z}] = 1+\sum_{r\geq1}^{\infty} \frac{1}{r!}\sum_{\substack{s_1, \dots , s_r\in {\mathcal{P}_{\E[Z]}}\\ s_i\sim s_j}} \prod _{u=1}^rz_{s_u}(\beta).
    \label{eq:polymer_firstmoment}
\end{equation}
where $z_{s_u}(\beta) = \frac{(-\beta \sigma)^{|s_u|}}{|s_u|!}\tr\left[\prod_{j\in [|s_u|]}\hat{\sigma}_{I_j}\right]$.  
\end{prop}
\begin{proof}[Proof sketch]
    It will be instructive to reproduce the beginning lines of the argument in \cite{zlokapa2026rigorousquasipolynomialtimeclassicalalgorithm} to motivate the origin of this polymer series. Consider moments
\begin{align*}
   \E[\tr[H^{m}]]
   &=\sum_{I_1,\ldots,I_{m}\in\binom{[N]}{q}\times\{X,Y,Z\}^q}
   \E[J_{I_1}\cdots J_{I_{m}}]\tr[\hat{\sigma}_{I_1}\cdots\hat{\sigma}_{I_{m}}] \\
   &=\sum_{I_1,\ldots,I_{m}}
   \sum_{\pi\in\Pi([m])}
   \prod_{\{a,b\}\in\pi}\E[J_{I_a}J_{I_b}]
   \tr[\hat{\sigma}_{I_1}\cdots\hat{\sigma}_{I_{m}}] \quad \text{By Wick's theorem} \\
   &=\sigma^{m}
   \sum_{I_1,\ldots,I_{m}}
   \sum_{\pi\in\Pi([m])}
   \prod_{\{a,b\}\in\pi}\delta_{I_a,I_b}
   \tr[\hat{\sigma}_{I_1}\cdots\hat{\sigma}_{I_{m}}] \\
   &=\sigma^{m}
   \sum_{\pi\in\Pi([m])}
   \sum_{\substack{I_1,\ldots,I_{m}\\
   I_a=I_b\;\text{for all }\{a,b\}\in\pi}}
   \tr[\hat{\sigma}_{I_1}\cdots\hat{\sigma}_{I_{m}}] \\
   &=\sigma^{m}\sum_{(\pi,\bm{I})\in\mathcal{S}_m}
   \tr[\hat{\sigma}_{I_1}\cdots\hat{\sigma}_{I_{m}}].
\end{align*}
Then it follows that 
\begin{equation}
    \E[\hat{Z}] =1+ \sum_{m\geq 1}\frac{1}{m!}\sum_{\tau \in \mathcal{S}_m } (-\beta\sigma)^m\tr[\hat{\sigma}_{I_1}\cdots\hat{\sigma}_{I_{m}}].
\end{equation}
The final observation comes from realizing that $\tau$ can be broken down into its connected components. In particular, suppose each tuple $\tau = (\pi, \textbf{I})$ has a support-intersection graph $G(\tau)$ with $r$ connected components $C_u, u\in [r]$. Pauli strings belonging to different components commute and the trace operator distributes over them, so we can decompose moments further:
\begin{equation}
    \E[\hat{Z}]= 1+ \sum_{m\geq 1}\frac{1}{m!}\sum_{\tau \in \mathcal{S}_m } \prod_{u=1}^{r}(-\beta\sigma)^{|\{j : j\in C_u\}|}
   \tr\left[\prod_{j\in C_u}\hat{\sigma}_{I_j}\right].
   \label{eq:closed_form_first_moment_untilted}
\end{equation}
This decomposition motivates the polymer representation: each connected component of $G(\tau)$ defines an irreducible contribution, which we identify with a polymer $s \in \mathcal{P}_{\E[Z]}$.  We refer to \cite[Lemma~14]{zlokapa2026rigorousquasipolynomialtimeclassicalalgorithm} where the formal bijection is derived between the tuple $\tau$ and a consistent tuple of $r$ polymers, up to a permutation of labels.  
\end{proof}

\begin{prop}[The second moment as a polymer expansion]\label{claim:polymer_set_second_moment}
The second moment of the quantum $p$-spin model can be expanded as a polymer series in terms of the second-moment polymer set $\mathcal{P}_{\E[Z^2]}$ of \cref{def:moment_polymer_sets}. That is,
\begin{equation}
    \E[\hat{Z}(\beta)\hat{Z}(\gamma)]
    =1+\sum_{r\geq1}^{\infty}\frac{1}{r!}
    \sum_{\substack{s_1,\dots,s_r\in\mathcal{P}_{\E[Z^2]}\\ s_i\sim s_j}}
    \prod_{u=1}^{r}z_{s_u}(\beta, \gamma),
    \label{eq:polymer_secondmoment}
\end{equation}
where
\begin{equation}
    z_{s_u}(\beta)
    =\frac{(-\beta\sigma)^{|\{j\in[|s_u|]\;:\;\mu(j)=1\}|}(-\gamma\sigma)^{|\{j\in[|s_u|]\;:\;\mu(j)=2\}|}}{|s_u|!}
    \tr\left[\prod_{\substack{j\in[|s_u|]\\ \mu(j)=1}}\hat{\sigma}_{I_j}\right]
    \tr\left[\prod_{\substack{j\in[|s_u|]\\ \mu(j)=2}}\hat{\sigma}_{I_j}\right].
\end{equation}
\end{prop}
We refer the reader to \cite[Lemma 20]{ zlokapa2026rigorousquasipolynomialtimeclassicalalgorithm} where a derivation of the analogous decomposition is given for the SYK model.

\subsubsection{Zero-freeness from polymer expansions}
The approach taken in \cite{zlokapa2026rigorousquasipolynomialtimeclassicalalgorithm} to show zero-freeness with high probability consists of three steps that make use of the polymer representations established above.

\paragraph{1. First-moment zero-freeness.}
First, it is shown that the mean $\E[Z]$ is nonzero. They prove this by showing that its polymer set satisfies the Kotecký--Preiss condition. We restate this condition for completeness. 
\begin{prop}[Kotecký--Preiss: Zero-freeness of the partition function]\label{prop:KP_first_moment}
 Let $\mathcal{P}$ be a polymer set with incompatibility relation $\nsim$ such that every polymer $s\in\mathcal{P}$ satisfies $s\nsim s$. For each finite $\Lambda\subset\mathcal{P}$ and vector $(z_s)_{s\in\mathcal{P}}\in\mathbb{C}^{\mathcal{P}}$, define the polymer partition function
   \begin{align}
      \Xi_\Lambda(z)
      &=1+\sum_{r\geq 1}\frac{1}{r!}
      \sum_{\substack{(s_1,\ldots,s_r)\in\Lambda^r\\
      s_j\sim s_k\ \text{for all }1\leq j<k\leq r}}
      \prod_{j=1}^{r}z_{s_j}. \label{eq:polymer-partition-function}
   \end{align}
Let $a:\mathcal{P}\to[0,\infty)$ be a function such that
   \begin{align}
      \sum_{s'\in\mathcal{P}:s'\nsim s}|z_{s'}|e^{a(s')}
      &\leq a(s)
   \end{align}
for all $s\in\mathcal{P}$. Then for every finite $\Lambda\subset\mathcal{P}$, one has
   \begin{align}
      \Xi_\Lambda(z) &\neq 0.
   \end{align}
\end{prop}

Next, it is shown that $Z$ concentrates about its mean so that typical instances $Z$ are bounded away from zero by a constant. To show concentration, the variance is shown to be small compared to the squared mean, $\Var[Z] \ll |\E Z|^2$ or more formally $\left |\frac{\E|Z|^2}{|\E Z|^2}-1\right |= o_N(1)$. It suffices to show that $\log\frac{\E[|Z|^2]}{|\E Z|^2}=o_N(1)$. The approach taken to establish this fact once more uses Kotecký--Preiss and is outlined in the next two steps.

\paragraph{2. Second-moment zero-freeness.}
The \textit{second-moment} polymer is also shown to satisfy the Kotecký--Preiss condition and be zero-free within a constant disk.

A consequence of this step is that a corollary of Kotecký--Preiss can be used:
\begin{corollary}[Pinned Kotecký--Preiss]\label{cor:pinned_KP}
   Let $\Lambda\subset\mathcal{P}$ be finite and assume $\Xi_\Lambda(z)$ is a polymer partition function satisfying \cref{prop:KP_first_moment}. Then for any $\Lambda_0\subseteq\Lambda$,
   \begin{align}
      \left|\log\frac{\Xi_\Lambda(z)}{\Xi_{\Lambda_0}(z)}\right|
      &\leq \sum_{s\in\Lambda\setminus\Lambda_0}|z_s|e^{a(s)}.
   \end{align}
\end{corollary}

\paragraph{3. Decay of the pinned tail.}
The final task involves showing the pinned KP tail decays with $N$, thereby completing the desired claim $\log\frac{\E{Z^2}}{\E|Z|^2}=o_N(1)$. 

We now prove that the partition function $Z = \Tr[e^{-\beta H}]$ of the quantum $p$-spin model is zero-free with high probability in a sufficiently small constant-size disk in the complex $\beta$ plane.

\begin{theorem}[Zero-free disk of the quantum $p$-spin partition function]\label{thm:zero_freeness}
With probability $1- o(n^{2-q})$ over the disorder of the quantum $p$-spin model, the partition function of the quantum $p$-spin model is zero-free $Z(\beta)\neq 0$ for all complex $\beta$ in the disk 
\begin{equation}
   J|\beta|\leq \frac{0.2256}{\sqrt{q}}.
\end{equation}
\end{theorem}

\subsubsection{Proof of \cref{thm:zero_freeness}}

We now apply the prescription for showing zero-freeness with high probability to our case of the quantum $p$-spin model.
\begin{lemma}
    $\E[Z] \neq 0$ in the disk $|\beta|\leq \frac{1}{J\sqrt{eq}}$
\end{lemma}
\begin{proof}[Proof sketch] 
The proof of this claim is largely adapted from \cite[Lemma~17]{zlokapa2026rigorousquasipolynomialtimeclassicalalgorithm}, with the following substitutions:
\begin{enumerate}
    \item $r = 0$: There are no perturbation insertions since we are considering the unperturbed partition function.  
    \item $|\tr[\prod_{j=1}^{2k}\hat{\sigma}_{I_j}]| =1$: The Wick-paired Pauli strings multiply to the identity $I$ up to a sign. 
    \item $\# \{\text{admissible } (K_1, \cdots K_k):\;x\in \bigcup_i K_i ) \} \leq 3^{qk}\binom{n-1}{q-1}^k(kq)^{k-1}$: For each node of the spanning tree in  part (1) of \cite[Lemma~12.2]{zlokapa2026rigorousquasipolynomialtimeclassicalalgorithm}, there exists an additional multiplicative factor of $3^q$ due to the additional Pauli terms. 
\end{enumerate}
Considered together, these substitutions produce 
\begin{align*}
\sum_{s'\in\mathcal{P}:s'\nsim s}|z_s(\beta)|e^{\frac{\supp(s')}{2q}}&\leq  \frac{\supp(s)}{q} T\left(
\frac{q \sqrt{e}}{2} 3^q
\binom{n-1}{q-1}\sigma^2 |\beta|^2
\right).\\
&\leq \frac{\supp(s)}{q} T\left(
\frac{q \sqrt{e}}{2} J^2 |\beta|^2
\right). \\
&\leq \frac{\supp(s)}{2q}\quad \text{for }J|\beta| \leq \frac{1}{\sqrt{eq}}
\end{align*}
Thus, the first-moment polymer set satisfies \cref{prop:KP_first_moment} for $\rho_s = |z_s(t)|$ and $\alpha(s) = \supp(s)/2q$, and the claim follows.
\end{proof}
Next, we show the analogous result for the second moment. The proof outlined here will be somewhat different than that in \cite[Section~5]{zlokapa2026rigorousquasipolynomialtimeclassicalalgorithm} since it is model-specific. 
\begin{lemma}\label{lemma:KP_second_moment}
$E[Z^2]\neq 0$ in the disk $|\beta| \leq\frac{0.2256}{J\sqrt{q}}$.
\end{lemma}
\begin{proof}
Using the polymer representation of $\E[Z^2]$ from \cref{claim:polymer_set_second_moment}, it suffices to establish \cref{prop:KP_first_moment}:
\begin{align}
    \sum_{s'\in\mathcal{P}_\text{mix}:s'\nsim s}|z_{s'}(t,\gamma)|e^{\frac{\supp(s')}{4q}}\leq\frac{\supp(s)}{4q}.
    \label{eq:KP_tail_second_moment}
\end{align}
First, we decompose the sum. Consider the tuple $\tau :=(\pi, \bm{I} , \mu)$ whose separated support-intersection graph $G^\text{sep}(\tau)$ is connected \textit{i.e.} $\tau \in \mathcal{P}^\text{sep}$. This forces $\mu$ to be constant because unless $\bm{I}$ lies entirely in either replica $1$ or $2$, $G^\text{sep}(\tau)$ would not be connected. Such a tuple automatically has a connected mixed support-intersection graph and also belongs to $\mathcal{P}^\text{mix}$. Thus, $\mathcal{P}^\text{sep} \subseteq \mathcal{P}^{\text{mix}}$. So it is admissible to write 

\begin{align*}
\sum_{s'\in\mathcal{P}_\text{mix}:s'\nsim s}|z_{s'}(t,\gamma)|e^{\frac{\supp(s')}{4q}} = \sum_{s'\in\mathcal{P}_\text{sep}:s'\nsim s}|z_{s'}(t,\gamma)|e^{\frac{\supp(s')}{4q}} + \sum_{s'\in\mathcal{P}_\text{mix}/ \mathcal{P}_\text{sep}:s'\nsim s}|z_{s'}(t,\gamma)|e^{\frac{\supp(s')}{4q}}.
\end{align*}

Implicitly, we have already upper-bounded the first term which corresponds to the \textit{square} of the first-moment KP tails since the sum factorizes over the two replicas. As in \cite[Corollary~22]{zlokapa2026rigorousquasipolynomialtimeclassicalalgorithm}, we take the (smaller) disk $\max(|\beta|, |\gamma|) \leq \frac{1}{\sqrt{2q\sqrt{e}}}$ to ensure the term is $\leq \frac{1}{2} \frac{\supp(s)}{2q}$.

Now we focus on the second term. We will show the stronger bound
\begin{equation}
    \sum_{s'\in\mathcal{P}_\text{mix}/ \mathcal{P}_\text{sep}}|z_{s'}(t,\gamma)|e^{\frac{\supp(s')}{4q}}\leq o_N(1).
\end{equation}
A couple of observations about the polymers $\tau=(\pi, \bm{I}, \mu) \in \mathcal{P}_\text{mix}/ \mathcal{P}_\text{sep}$ are in order. Notice that a Wick-pair $(a,b) \in \pi$ can be either \textit{same-replica} ($\mu(a) =\mu(b)$) or \textit{cross-replica} ($\mu(a)\neq \mu(b)$).  Let $k$ denote the total number of Wick pairs of which $s$ are cross-replica and $m = k - s$ same-replica. Then, we may reparameterize the sum as
\begin{align*}
    \sum_{\tau\in\mathcal{P}_\text{mix}/ \mathcal{P}_\text{sep}}|z_\tau(t,\gamma)|e^{\frac{\supp(\tau)}{4q}} &= \sum _{s=1}^{\infty}\sum_{m=0}^{\infty}M_{m,s}\qquad\text{where }M_{m,s}
:=
\sum_{\substack{
\tau\in\mathcal{P}_{\mathrm{mix}}\setminus\mathcal{P}_{\mathrm{sep}}\\
\tau\text{ with $m$ same-replica} \\
e\text{ cross-replica}
}}
\end{align*}
The polymers contributing to the sum above must satisfy 
\begin{equation}
    |z_{\tau}(t,\gamma)| \propto 
    \tr\left[\prod_{\substack{j\in[|s'_u|]\\ \mu(j)=1}}\hat{\sigma}_{I_j}\right]
    \tr\left[\prod_{\substack{j\in[|s'_u|]\\ \mu(j)=2}}\hat{\sigma}_{I_j}\right] \neq 0.
    \label{eq:nonvanishing_untilted}
\end{equation} These are precisely polymers for which the product of Pauli strings within each replica is proportional to the identity. Within each replica, a Pauli string either belongs to a same-replica pair, in which case its partner occurs also in the same replica and together they multiply to the identity, \textit{or} it is one half of a cross-replica pair. The condition that the traces be nonvanishing imposes a constraint within each replica that remaining cross-replica halves  $\mathcal{W} := \{\hat{\sigma}_1, \dots, \hat{\sigma}_s\}$ also multiply to the identity. This immediately rules out $s=1$ since a single Pauli string is traceless.
Denote by $v$ the size of the support of this set, $v:= |\bigcup_{\hat{\sigma}_I \in \mathcal{W}} \supp(\hat{\sigma}_I)|$. The sum above can be further parameterized in terms of $v$. It is clear that $q\leq v$ which is the size of any one Pauli string. We also have the trivial upper bound $v\leq N$. But this can be refined further by noting that the multiply-to-identity constraint on $\mathcal{W}$ means that each of the $v$ sites must be touched at least twice by the Pauli strings in $\mathcal{W}$. Thus, $v\leq \frac{sq}{2}$.  So, 
    \begin{align*}
    \sum_{s\in\mathcal{P}_\text{mix}/ \mathcal{P}_\text{sep}}|z_s(\beta,\gamma)|e^{\frac{\supp(s)}{4q}} &= \sum_{s=1}^{\infty}\sum_{m=0}^{\infty}\sum_{v=q}^{\min(N, sq/2)}M_{m,s, v}
    \end{align*}
We focus on upper-bounding the summand $M_{m,s,v}$.  
An arbitrary polymer $\tau \in $  with Wick-pair structure $(m,s,v)$ satisfies
\begin{align*}
    \supp(\tau) \leq \underbrace{2v}_{\text{cross-replica }} + \underbrace{mq}_{\text{same-replica}}
\end{align*}
and thus the Kotecký--Preiss weight term is upper-bounded by 
\begin{equation}
    |z_s(\beta,\gamma)|e^{\frac{\supp(s)}{4q}}\leq \frac{|\beta\sigma|^{2k}}{(2k)!}e^{2v + mq}
\end{equation}
assuming that without loss of generality, $|\beta| \geq |\gamma|$. Next we bound the number of polymers with Wick-pair structure $(m,s,v)$.  The total number of labels consistent with this structure is
\begin{align*}
    \Omega_{\text{consistent labelings}} 
    &\leq \text{\#pairings of $2k$ elements}\\
    &\quad\times \text{choice of $s$ cross-replica pairs}\\
    &\quad\times \text{same-replica assignments}\\
    &\quad\times \text{cross-replica endpoint assignments}\\
    &= (2k-1)!! \binom{k}{e} 2^e 2^m \\
    &= \frac{(2k)!}{e!m!}
\end{align*}
Next, we do count the possible sets $\mathcal{W}$ consisting of $s$ cross-replica Pauli strings whose joint support is of size $v$:
\begin{align}
    \Omega_{\text{cross-replica set}}\leq  \binom{N}{v}\left(3^q \binom{v}{q}\right)^e .
\end{align}
For a fixed cross-replica set $\mathcal{W}$, we count the choices of $m$ same-replica Pauli strings $\{\sigma_{I_1},\ldots,\sigma_{I_m}\}$ for which the resulting mixed-support intersection graph $G^{\mathrm{mix}}(\tau)$ is connected. This factor is a version of \cite[Lemma~12.2]{zlokapa2026rigorousquasipolynomialtimeclassicalalgorithm}, adapted to the case of $\mathcal{W}$:
\begin{align*}
    \Omega_{\text{consistent same-replica sets}} \leq v \bigg(3^q \binom{N-1}{q-1}\bigg)^m(v+mq)^{m-1}, \quad m\geq1.
\end{align*}
and for $m=0$, this count is trivially one. 
Putting things together, we obtain 
\begin{align*}
M_{m,s,v}&= 
\sum_{\substack{
s\in\mathcal{P}_{\mathrm{mix}}\setminus\mathcal{P}_{\mathrm{sep}}\\ \text{ consistent with $(m,s,v)$}}} 
|z_s(\beta,\gamma)|e^{\frac{\supp(s)}{4q}}\\
&\leq    \frac{|\beta\sigma|^{2(s+m)}}{s!m!}\Omega_{\text{cross-replica set}}
    v\left(3^q\binom{N-1}{q-1}\right)^m(v+mq)^{m-1}
    e^{\frac{2v+mq}{4q}}\\
    &\leq \frac{\Omega_{\text{cross-replica set}}}{s!}(\sigma^2|\beta|^2)^s e^{v/(2q)}
    \frac{v(v+mq)^{m-1}}{m!}
    \left[\underbrace{3^q\binom{N-1}{q-1}\sigma^2}_{J^2}|\beta|^2e^{1/4}\right]^m.
\end{align*}
Note that the $m=0$ case simply has $\M_{(0,s,v)}\leq\frac{\Omega_{\text{cross-replica set}}}{s!}(\sigma^2|\beta|^2)^s e^{v/(2q)}$.
Next, we sum over $m$. For clarity, set 
$a:= \frac{v}{q}$ and $x:=qJ^2|\beta|^2e^{\frac14}$. Then,
\begin{align*}
    \sum_{m\geq0}M_{m,s,v} &\leq   \frac{\Omega_{\text{cross-replica set}}}{s!}(\sigma^2|\beta|^2)^s s^{v/(2q)}\left(1+ \sum_{m\geq1}
    \frac{a(a+m)^{m-1}}{m!}
    x^m\right).\\
    &\leq  \frac{\Omega_{\text{cross-replica set}}}{s!}(\sigma^2|\beta|^2)^s \exp\left[\frac{v}{2q} + \frac{v}{q}T(x)\right]\qquad \text{$T(x)$ is the standard Tree function. }\\
    &\leq \frac{\Omega_{\text{cross-replica set}}}{s!}(\sigma^2\underbrace{|\beta|^2\exp\left[\frac{1}{4} + \frac{1}{2}T(x)\right]}_{B(\beta)})^s.
\end{align*}
We furthermore define $u=\sigma^2 B(\beta)$. For reasons that will become clear later, we consider two regimes.
\paragraph{Case 1: $2\leq s<\frac{2n}{q}$:}We will use the following upper-bound on $\Omega_{\text{cross-replica set}}\leq \binom{n}{v}\left(3^q \binom{v}{q}\right)^s$. Thus,
\begin{equation}
        M_{s,v} \leq \frac{u^s}{s!}\binom{N}{v}\left(3^q \binom{v}{q}\right)^s .
\end{equation}
Using standard binomial bounds
\begin{align*}
        M_{s} &\leq \frac{u^s}{s!}\sum_{v\leq \min(eq/2,n)}\left(\frac{eN}{v}\right)^{v}\left(\frac{3ev}{q}\right)^{qs} .
\end{align*}
The summand is an increasing function of $v$ as it has a positive derivative $\forall v\leq N$. Therefore, we evaluate it at $v=sq/2$ to obtain
\begin{align*}
        M_{s} &\leq \frac{u^s}{s!}\frac{sq}{2}\left(\underbrace{\left(\frac{2e}{q}\right)^{q/2}\left(\frac{3e}{2}\right)^{q}}_{:=C}N^{q/2}s^{q/2}\right) ^s.\\
        &\leq \frac{sq}{2}\frac{\left(CuN^{q/2}s^{q/2}\right) ^s}{s!}
\end{align*}
Using $1/s! \leq (s/e)^s$,
\begin{align}
     M_s &\leq \frac{sq}{2}\left(eCuN^{q/2}s^{q/2-1}\right) ^s
\end{align}
Since $\sigma^2 \leq \frac{q^{q-1}}{3^qn^{q-1}}$, 
\begin{align}
     M_s &\leq \frac{sq}{2}\left(eC\frac{q^{q-1}}{3^q}B(\beta)\left(\frac{s}{N}\right)^{q/2-1}\right) ^s
\end{align}
Notice that the ratio
\begin{align*}
\frac{M_{s+1}}{M_{s}} &\leq 2e^{q/2-1}\left(eC\frac{q^{q-1}}{3^q}B(\beta)\right)\left(\frac{s+1}{N}\right)^{q/2 -1}\\
&\leq  2e^{q/2-1}\left(eC\frac{q^{q-1}}{3^q}B(\beta)\right)\left(\frac{2}{q}\right)^{q/2 -1}\\
&\leq e^{2q} B(\beta)\\
&\leq \frac12 \quad \text{for $|\beta|\leq \frac{0.2256}{J\sqrt{q}}$} \
\end{align*}
Then the target sum is geometric:
\begin{align*}
    \sum_{s=2}^{2n/q}M_s &\leq 2M_2\\
    &\leq q2^{q-1}\left(eC\frac{q^{q-1}}{3^q}B(\beta)\right)^2N^{2-q} = O(N^{2-q})
\end{align*}
\paragraph{Case 2: $s\geq2n/q$}
In this regime, it will suffice to use the looser bound $\Omega_{\text{cross-replica set}}\leq \left(3^q \binom{N}{q}\right)^s$ so that
\begin{align*}
    M_s &\leq \frac{1}{s!} \left(\sigma^2 B(\beta) 3^q \binom{N}{q}\right)^2\\
    &\leq \frac{1}{s!}\left(\frac{N}{q}B(\beta)\right)^s\\
    &\leq \left(\frac{Ne}{sq}B(\beta)\right)^s\\
    &<\left(\frac{2N}{sq}\right)^s
\end{align*}
Where in the last line we have used that under the same temperature constraint as in the previous case $eB(\beta)/2 <1$. Now consider the ratio
\begin{align*}
    \frac{M_{s+1}}{M_s} = B(\beta)\frac{N/q}{s+1}\leq \frac{B(\beta)}{2} \quad \text{since }s\geq\frac{2N}{q}
\end{align*}
Since $M_{2N/q} <1$, the sum is geometric:
\begin{align*}
    \sum_{s\geq 2N/q}M_s  \leq \frac{(eB(\beta)/2)^{2N/q}}{1-B(\beta)/2} \leq \exp({-\Omega(N)})
\end{align*}
Putting it all together, 
\begin{equation}
    M= \sum_s M_s \leq \exp({-\Omega(N)}) + O(N^{2-q}) = o(1)
    \label{eq:decaying_KP}
\end{equation}
The statement then follows by applying \cite[Lemma~25]{zlokapa2026rigorousquasipolynomialtimeclassicalalgorithm} to the tail in \eqref{eq:KP_tail_second_moment}.
\end{proof}

 \cref{cor:pinned_KP} now applies to the second-moment polymer series. We choose $\Lambda = \mathcal{P}_\text{mix}$ and $\Lambda_0 = \mathcal{P}_\text{sep}$. Then, 
\begin{align*}
      \left|\log\frac{\E|Z(\beta)|^2}{|\E Z(\beta)|^2} \right|&= \left|\log\frac{\Xi_\Lambda(z)}{\Xi_{\Lambda_0}(z)}\right|\leq  \sum_{s\in\Lambda\setminus\Lambda_0}|z_s(\beta, \bar{\beta})|e^{a(s)} = O(N^{2-q}).
\end{align*}
where the last line follows from the \cref{eq:decaying_KP}. Then applying \cite[Section~5.3]{zlokapa2026rigorousquasipolynomialtimeclassicalalgorithm}, gives the desired probability of success.


\subsection{Conjectured w.h.p. zero-freeness of the tilted partition function}\label{sec:conj_conc}

In the previous section, we showed zero-freeness for the \textit{untilted} partition function of the quantum $p$-spin model $\tr[e^{-\beta H}]$. This result is not sufficient for estimating the magnetizations of the untilted model as the finite difference method in \cref{claim:mag_from_finite_difference}  additionally requires the zero-freeness of the tilted partition function  $\tr[e^{-t\sigma^z_i }e^{-\beta H}]$. In fact, since ASL requires an efficient magnetization estimator at any tilt along its path, in general we require zero-freeness for any tilted $\hat{Z}_y :=\tr[e^{-\sum_{i=1}^{N} y_i \sigma_k^z} e^{-\beta H}]$. 

It is easy to see that the first-moment results for the untilted case port over to the tilted case because of the following Proposition:
\begin{prop}\label{prop:tilt_exp_zero_free}
    $E[\hat{Z_y}] \neq 0 \iff E[\hat{Z}]\ne 0$
\end{prop}
\begin{proof}
    Notice that by the Wick-pairing condition, the Pauli strings in \cref{eq:closed_form_first_moment_untilted} must multiply to the identity up to a constant $\omega_{\tau}$:
    \begin{align*}
        \E[e^{-\beta H}] &=\left(\sum_{m\geq 0}\frac{1}{m!}\sum_{\tau \in \mathcal{S}_m } (-\beta\sigma)^m\omega_\tau \right)I = \alpha(\beta)I\\
    \end{align*}  
    Where $\alpha(\beta)$ is some scalar function. Thus, $\E[\hat{Z}] = \alpha(\beta)$. Now consider the tilted partition function
    \begin{align*}
        \E[\hat{Z_y}]&=\tr[e^{-\sum_{i=1}^{N} y_i \sigma_k^z} \E[e^{-\beta H}]]\\
        &=\left(\prod_{i=1}^{N}2\cosh(y_i)\right)\E[\hat{Z}]
    \end{align*}
    The $\cosh$ prefactor is nonzero $\forall \bm{y}\in \R^N$ and thus the claim follows. 
    \end{proof}

Showing concentration of the tilted partition function is, however,  more complicated. As in the untilted case, we may similarly derive a polymer expansion.
\begin{prop}[The tilted first moment as a polymer expansion]
Given a tilt $\bm{y}\in\R^N$, the first moment of the tilted quantum $p$-spin partition function can be expanded as a polymer series in terms of a tilted first-moment polymer set $\mathcal{P}_{\E[Z_{\bm{y}}]}$ with compatibility relation $\sim$. That is,
\begin{align}
    \E[Z_{\bm{y}}(\beta)] &= \prod_{i=1}^{N}\cosh(y_i)\left(1+\sum_{r\geq1}^{\infty}\frac{1}{r!}\sum_{\substack{s_1,\dots,s_r\in\mathcal{P}_{\E[Z_{\bm{y}}]}\\s_i\sim s_j}}\prod_{u=1}^r z_{s_u}(\beta,\bm{y})\right),
    \label{eq:polymer_tilted_firstmoment}
\end{align}
where, for a polymer $s=(\pi,\vec{I})\in\mathcal{P}_{\E[Z_{\bm{y}}]}$, we define
\begin{align}
    z_s(\beta,\bm{y}) &= \frac{(-\beta\sigma)^{|s|}}{|s|!}\tr\left[\prod_{k\in\supp(s)}\frac{e^{-y_k\hat{\sigma}_k^z}}{\cosh(y_k)}\prod_{j=1}^{|s|}\hat{\sigma}_{I_j}\right].
\end{align}
\end{prop}
\begin{proof}
Note that the tilt operator can be expressed as a product of single site operators, i.e. 
\begin{equation}
    e^{-\sum_{k=1}^{N} y_k \sigma_k^z} = \otimes_{k=1}^{N} e^{-y_k \sigma^z_k}
\end{equation} 
We then commute every factor through each connected component until it hits the unique connected component whose support it is a member of. In this way, each connected component has appended to it all the factors of the tilted operator that comprise its support. Since the components are still mutually disjoint, the trace still distributes over the disjoint subspaces and we obtain the formula above.

\end{proof}
However, the Kotecký--Preiss condition of the previous section used to show concentration breaks down for the polymer series above. To see why, notice that the nonvanishing trace condition of \cref{eq:nonvanishing_untilted} is weakened by the presence of the tilt. In the untilted case, a site which appears only once in a replica necessarily leaves a nonidentity Pauli operator at that site and hence gives zero trace. By contrast,
\begin{align*}
    \frac{e^{-y_i\sigma_i^z}}{\cosh y_i} &= I-\tanh(y_i)\sigma_i^z,\\
    \tr\left[\frac{e^{-y_i\sigma_i^z}}{\cosh y_i}\sigma_i^z\right] &= -\tanh(y_i),
\end{align*}
which is generically nonzero. Thus, a site may now appear only once, provided the corresponding Pauli is $\sigma_i^z$.

In particular, in the mixed second-moment expansion, a polymer containing only a single mixed Wick pair (\textit{i.e.} $s=1$) need no longer vanish. If the common Pauli string is a $Z$-string supported on some $I\in\binom{[N]}{q}$, then each replica contains precisely one copy of this string and nevertheless has nonzero tilted trace. There are $\binom{N}{q}$ such choices, while the mixed Wick pair contributes only a factor $\sigma^2=\Theta(N^{-(q-1)})$. Hence the total contribution of these single-mixed-pair terms scales as
\begin{align*}
    \binom{N}{q}\sigma^2 &= \Theta(N).
\end{align*} 
Thus, the mixed-polymer Kotecký--Preiss tail is no longer $o(1)$ as its lowest-order contribution already diverges with $N$.

While widely believed to be true, even in the classical spin glass setting provably demonstrating the zero-freeness of the partition function in the presence of a generic tilt (which for classical models is equivalent to an external field) remains an open problem  \cite[Open Problem 3]{bencs_zeros_2026}. Inspired by this, we thus make the following conjecture:

\begin{conj}[Zero-freeness of tilted partition function; quantum analogue of Open Problem 3 in \cite
{bencs_zeros_2026}]\label{claim:para_zero_free}
 There exists $\delta>0$ and $\beta^\ast>0$ such that w.h.p., for all $\bm{y}\in\mathbb{R}^N$, $Z_{\bm{y}}\left(\beta\right)$ is zero-free whenever $\Re\left\{\beta\right\}\leq\beta^\ast$ and $\Im\left\{\beta\right\}\leq\delta$. 
\end{conj}
We motivate this conjecture by recalling Lee-Yang theory, which claims that a thermodynamic phase transition at an inverse temperature $\beta^*$ is signaled by the complex zeros of $Z(\beta)$ approaching the real axis at exactly $\beta^*$ in the thermodynamic limit \cite{blythe_lee-yang_2003}. Thus, a zero-free strip (or generally, any region containing the real $\beta$ line) extending to some $\beta$ is consistent with the absence of a thermodynamic phase transition up to that temperature. In this picture, \cref{thm:zero_freeness} implies that the untilted system remains in the replica-symmetric phase at least for $\beta < \frac{0.2256}{J\sqrt{q}}$. If we view the tilted system as the untilted system subject to an external field $\bm{y}$, it remains to argue that
an external field only maintains the replica-symmetric phase. Indeed, for classical spin glasses this is known to be true. Standard results on the phase diagram of the classical Sherrington--Kirkpatrick model in the presence of an external field \cite{guerra_replica_2006, bena_statistical_2005} reveal that the spin-glass phase is bounded by the de Almeida--Thouless line in the temperature--field plane. This phase diagram indicates that an external field only suppresses the spin-glass phase and eventually drives the system into the replica-symmetric phase. In this sense, an external field acts to remove, rather than induce, spin-glass order. This provides heuristic support for the conjecture that introducing a tilt should retain a zero-free strip already admitted by the untilted model.

\printbibliography

\end{document}